\documentclass{article}

\usepackage[utf8]{inputenc}
\usepackage{amsfonts,amsmath,amssymb,amsthm}
\usepackage{verbatim,float,url,enumerate}
\usepackage[format=plain,
            font=it]{caption} 
\usepackage{graphicx,subcaption,epsfig,psfrag}
\usepackage{dsfont,bm,color,appendix}
\usepackage{centernot}
\usepackage{natbib}
\setcitestyle{authoryear, open={(},close={)}}
\usepackage[final]{pdfpages}
\usepackage{fancyhdr}
\usepackage{tcolorbox}
\usepackage[margin=1in, lmargin = 1in, rmargin = 1in]{geometry}
\usepackage{pdflscape} 
\usepackage{hyperref}
\usepackage{multirow,multicol}
\usepackage{tabularx}
\usepackage{arydshln}

\usepackage{algorithm}
\usepackage[noend]{algpseudocode}
\newtheorem{theorem}{Theorem}
\newtheorem{lemma}{Lemma}
\newtheorem{remark}{Remark}
\newtheorem{corollary}{Corollary} 
\newtheorem{proposition}{Proposition}
\newtheorem{definition}{Definition}

\newcommand{\argmin}{\mathop{\mathrm{argmin}}}

\usepackage{array,etoolbox}
\preto\tabular{\setcounter{magicrownumbers}{0}}
\newcounter{magicrownumbers}

\def\E{\mathbb{E}}

\def\Var{\mathrm{var}}

\def\tr{\mathrm{tr}}

\def\R{\mathbb{R}}

\def\Ical{\mathcal{I}}

\def\cN{\mathcal{N}}
\def\cI{\mathcal{I}}

\def\E{\mathbb{E}}
\def\simiid{\stackrel{i.i.d.}{\sim}}
\newcommand{\indep}{\perp \!\!\! \perp}
\def\eqd{\stackrel{d}{=}}
\newcommand{\norm}[1]{\left\lVert#1\right\rVert}

\def\err{\mbox{\textnormal{Err}}}
\def\errhat{\widehat{\mbox{\textnormal{Err}}}}

\def\cor{\textnormal{cor}}

\def\xy{{XY}}
\def\inn{{(\textnormal{in})}}
\def\out{{(\textnormal{out})}}
\def\train{{(\textnormal{train})}}
\def\test{{(\textnormal{out})}}
\def\ncv{{(\textnormal{NCV})}}
\def\cv{{(\textnormal{CV})}}
\def\inn{{\textnormal{in}}}
\def\splitt{{(\textnormal{split})}}
\def\cp{{(\ensuremath{C_p})}}
\def\rcp{{(\ensuremath{RC_p})}}
\def\A{\mathcal{A}}
\def\mse{\textsc{MSE}}
\def\msehat{\widehat{\mse}}
\def\sehat{\widehat{\textsc{se}}}

\def\oob{{\textnormal{(OOB)}}}
\def\boot{{\textnormal{(.632)}}}

\newcommand{\var}{\textnormal{var}}
\newcommand{\bias}{\textnormal{bias}}
\newcommand{\cov}{\textnormal{cov}}

\title{Cross-validation: what does it estimate and how well  does it do it?}

\author{Stephen Bates\thanks{Depts. of Statistics and EECS, Univ. of California, Berkeley; stephenbates@berkeley.edu}, Trevor Hastie\thanks{Depts. of Statistics and Biomedical Data Science, Stanford Univ.; hastie@stanford.edu}, and Robert Tibshirani\thanks{Depts. of Biomedical Data Science and Statistics,
    Stanford Univ.; tibs@stanford.edu}}
\date{\today}

\begin{document}
\maketitle

\begin{abstract}
Cross-validation is a widely-used technique to estimate prediction error, but its behavior is complex and not fully understood. Ideally, one would like to think that cross-validation estimates the prediction error for the model at hand, fit to the training data.  We prove that this is not the case for the linear model fit by ordinary least squares; rather it estimates the average prediction error of models fit on other unseen training sets drawn from the same population.  We further show that this phenomenon occurs for most popular estimates of prediction error, including data splitting,  bootstrapping, and Mallow's $C_p$. 
Next, the standard  confidence intervals for prediction error derived from cross-validation may have coverage far below the desired level. Because each data point is used for both training and testing, there are correlations among the measured accuracies for each fold, and so the usual estimate of variance is too small.  We introduce a nested cross-validation scheme to estimate this variance more accurately, and we show empirically that this modification leads to intervals with approximately correct coverage in many examples where traditional cross-validation intervals fail.
\end{abstract}

\section{Introduction}

When deploying a predictive model, it is important to understand its prediction accuracy on future test points, so both good point estimates and accurate confidence intervals for prediction error are essential. Cross-validation (CV) is a widely-used approach for these two tasks, but in spite of its seeming simplicity, its operating properties remain opaque. Considering first estimation, it turns out be  challenging to precisely state the estimand corresponding to the cross-validation point estimate. In this work, we show that the the estimand of CV is not the accuracy of the model fit on the data at hand, but is instead the average accuracy over many hypothetical data sets. Specifically, we show that the CV estimate of error has larger mean squared error (MSE) when estimating the prediction error of the final model than when estimating the average prediction error of models across many unseen data sets for the special case of linear regression.
Turning to confidence intervals for prediction error, we show that na\"ive intervals based on CV can fail badly, giving coverage far below the nominal level; we provide a simple example soon in Section~\ref{subsec:harrell_model}.
The source of this behavior is the estimation of the variance used to compute the width of the interval: it  does not account for the correlation between the error estimates in different folds, which arises because each data point is used for both training and testing. As a result, the estimate of variance is too small and the intervals are too narrow.
To address this issue, we develop a modification of cross-validation, \emph{nested cross-validation} (NCV), that achieves coverage near the nominal level, even in challenging cases where the usual cross-validation intervals have miscoverage rates two to three times larger than the nominal rate. 

\subsection{A simple illustration}
\label{subsec:harrell_model}
As a motivating example where na\"ive cross-validation confidence intervals fail, we consider a sparse logistic regression model 
\begin{equation}
\label{eq:log_regression_def}
P(Y_i = 1 \mid X_i = x_i) = \frac{1}{1 + \exp\{-x_i^\top \theta\}} \qquad i = 1,\dots,n,
\end{equation}
with $n=90$ observations of $p=1000$ features, and a coefficient vector $\theta = c \cdot (1, 1, 1, 1, 0, 0, \dots)^\top \in \mathbb{R}^p$ with four nonzero entries of equal strength. The feature matrix $X \in \mathbb{R}^{n \times p}$ is comprised of independent and identically distributed (i.i.d.) standard normal variables, and we chose the signal strength $c$ so that the Bayes misclassification rate is 20\%. 
We estimate the parameters using $\ell_1$-penalized logistic regression with a fixed penalty level. In this case, na\"ive confidence intervals for prediction error are far too small: intervals with desired miscoverage of 10\% give 31\% miscoverage in our simulation. We visualize this in Figure~\ref{fig:var_inflation_viz}. The intervals need to be made larger by a factor of about $1.6$ to obtain coverage at the desired level in this case. 

\begin{figure}[t]
\begin{center}
\includegraphics[height = 2.5in]{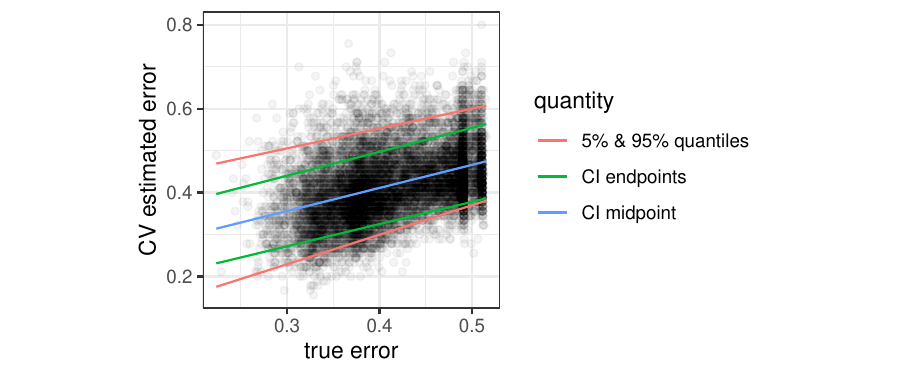}
\end{center}
\caption{A plot of the true error of a model versus the CV estimates for 1000 replicates of the model from Section~\ref{subsec:harrell_model}. The blue curve shows the average midpoint of the na\"ive CV confidence intervals. The green bands show the average 90\% confidence interval for prediction error given by na\"ive CV. The red curves show the  5\% and 95\% quantiles from a quantile regression fit. To achieve nominal coverage, the green curves should approximate the red curves, but they are too narrow in this case.}
\label{fig:var_inflation_viz}
\end{figure}

\subsection{Related work}
\label{subsec:related_work}

Cross-validation is used ubiquitously to estimate the prediction error of a model \citep{Allen74,geisser1975,stone1977}. 
The enduring popularity of CV is due to the fact that it is a conceptually simple improvement over a one-time train-test split \citep{Blum1999BeatingTH}. 
CV is part of a broader landscape of resampling techniques to estimate prediction error, with bootstrap-based techniques as the most common alternative \citep{efron1983estimating, efron1986, Efron1997, efron1993introduction}. The other main category of prediction error estimates are based on analytic adjustments such as Mallow's $C_p$ \citep{mallows1973comments}, AIC \citep{akaike1974aic}, BIC \citep{schwarz1978bic}, and general \emph{covariance penalties} \citep{stein1981estimating,efron2004estimation}. The present work is primarily concerned with CV, but also addresses the properties of bootstrap, data splitting and covariance penalty methods. 

In spite of CV's apparent simplicity, the formal properties of this procedure are subtle; the seemingly basic question  ``what is cross-validation estimating?''  has engendered considerable debate. Although the predictive accuracy of the model fit on the observed training data may seem like  a natural estimand, it has been observed that the CV estimator tracks this quantity only weakly, suggesting that CV should instead be treated as an estimator of the average prediction error across training sets \citep{zhang1995, hastie01statisticallearning, yousef2020leisurely}. See also \citet{Rosset2020} and \citet{Wager2019} for a discussion about different potential estimands. In this work, we discuss this phenomenon in detail for the case of the linear model. Our main result uses a conditional independence argument to explain the aforementioned weak relationship between CV and the instance-specific error.

Turning to the question of inference, one important use of CV is to deliver confidence intervals for the prediction error (or, similarly, an estimate of the standard error) to accompany a point estimate. The second primary goal in this work is to provide such confidence intervals, which cannot be reliably created with na\"ive methods, as shown in our example in Figure~\ref{fig:var_inflation_viz}. A fundamental prior result shows that there is no unbiased estimator of the standard error of the CV point estimate based on one instance of CV \citep{bengio2004no}. As a result, to obtain standard error estimates, one would either need to modify the CV procedure or make additional assumptions. Pursuing the former, \citet{Dietterich1998} and \citet{nadeau2003} proposes sampling schemes where the data is split in half, and CV is carried out within each half separately. This yields an estimate of standard error, but it will typically be much too conservative since the internal CV model fits each use a samples size that is less than half of the full sample. A related proposal due to \citet{Austern} involves repeatedly performing leave-one-out CV with data sets of half of the original size, but this proposed estimator is not computationally feasible for most learning algorithms.

In a different direction, \citet{nadeau2003} and \citet{Markatou2005} propose alternative estimates of standard error, but these are based only on the sample size and higher moments of the errors and so do not address the source of the problem: a covariance term  that we describe in Section~\ref{sec:CV_cov_structure}. For bootstrap estimators, there are proposals to estimate the standard error of the (bootstrap) point estimates of prediction error with methods based on influence functions \citep{efron1983estimating, Efron1997}, and this approach has been partially extended to CV for the special case of the area under the curve measure (AUC) of performance \citep{yousef2019}. The CV proposal of \citet{Austern} similarly involves leave-one-out resampling, which can be interpreted as an empirical estimate of the influence functions.

Accompanying these algorithmic proposals, there is some theoretical understanding of the asymptotic behavior of CV. \citet{dudoit2005} proves a central limit theorem (CLT) for a cross validation estimator, although it does not come with an estimate of the standard error. \citet{ledell2015} provides a consistent estimator for the standard error in the special case of estimating the AUC, and \citet{Benkeser2020improved} conducts a higher-order asymptotic analyses for AUC and other metrics, yielding a more efficient estimator for accuracy with a consistent standard error estimate. Further theoretical results establish the asymptotic normality of the CV estimate in more general cases \citep{Austern, bayle2020}. The former considers the average prediction error across training sets (similar to our goal herein), and introduces an asymptotically valid estimate of the standard error; see Appendix~\ref{subapp:az_se_estimator} for an experiment with this estimator.
The latter estimates a different estimand: the average prediction error of the models fit on the subsamples, and introduces a valid estimate of standard error for this quantity. We explain this estimand and the proposed standard error estimator in more detail in Appendix~\ref{app:kfold_test_error}.
Both use arguments relying on notions of algorithmic stability \citep{pmlr-v28-kumar13a, kale2011, celisse2016stability}. 
At present, it is not clear how the large-sample regime considered in these works relates to the behavior we see in small samples such as in the experiment in Section~\ref{subsec:harrell_model}.In particular, algorithmic stability may not be satisfied in high-dimensional settings or with small sample sizes; see Appendix~\ref{app:kfold_test_error} and~\citet{bayle2020} for more discussion.

Lastly, we note that CV is often used to compare predictive models, such as when  selecting a model or a good value of a learning algorithm's hyperparameters \citep[e.g.,][]{stoica1986model, shao1993linear, zhang1993model, Dietterich1998, xu2001, yang2007consistency, arlot2010, Fong2020marginal, sivula2020uncertainty, riley2021penalization}. To this end, \cite{varma2006bias} and \citet{Tibshirani2009} suggest a bias-correction for the model selected by cross-validation, \citet{Lei2019} shows how to return a confidence set for the best model parameters, and \citet{yang2007consistency, Wager2019} show that for CV, comparing two models is a statistically easier task than estimating the prediction error, in some sense. While we expect that our proposed estimator would be of use for hyperparameter selection because it yields more accurate confidence intervals for prediction error,  we do not pursue this problem further in the present work.

\subsection{Our contribution}
This work has two main thrusts. First, we study the choice of estimand for CV, giving results for the special case of the linear model. We prove a finite-sample conditional independence result (Theorem~\ref{thm:ols_cv_independent_err}) with a supporting asymptotic result (Theorem~\ref{thm:err_and_err_x_higd}) that together show that CV does not estimate the error of the specific model fit on the observed training set, but is instead estimating the average error over many training sets (Corollary~\ref{cor:cor_errx_errxy_highd} and Corollary~\ref{cor:err_errxy_gap_highd}). We also show that this holds for the other common estimates of prediction error: data splitting (Section~\ref{subsec:data_splitting}), Mallow's $C_p$ (Section~\ref{subsec:cov_penalties}), and bootstrap (Appendix~\ref{app:bootstrap_theorems}).
Second, we introduce a modified cross-validation scheme to give accurate confidence intervals for prediction error.  We prove that our estimate for the MSE of the CV point estimate is unbiased (Theorem~\ref{thm:ncv_unbiased}). Moreover, we validate our method with extensive numerical experiments, confirming that the coverage is consistently better than that of standard cross-validation (Section~\ref{sec:experiments}).

\section{Setting and notation}

We consider the supervised learning setting where we have features $X = (X_1,\dots,X_n) \in \mathcal{X}^n$ and response $Y = (Y_1,\dots,Y_n) \in \mathcal{Y}^n$, and we assume that the data points $(X_i, Y_i)$ for $i=1,\dots,n$ are i.i.d. from some distribution $P$.  We wish to understand how well fitted models  generalize to unseen data points, which we formalize with a loss function 
$$
\ell(\hat{y}, y): \mathcal{Y} \times \mathcal{Y} \to \R_{\ge 0}
$$ 
such that $\ell(y, y) = 0$ for all $y$. For example, $\ell$ could be squared error loss, misclassification error, or deviance (cross-entropy). Now consider a class of models parameterized by $\theta$. Let $\hat{f}(x, \theta)$ be the function that predicts $y$ from $x \in \R^p$ using the model with parameters $\theta$, which takes values in some space $\Theta$. Let $\A$ be a model-fitting algorithm that takes any number of data points and returns a parameter vector $\hat{\theta} \in \Theta$.
Let $\hat{\theta} = \A(X, Y)$ be the fitted value of the parameter based on the observed data $X$ and $Y$. We are interested in the {\em out-of-sample error} with this choice of parameters: 
$$
\err_{\xy} := \E \left[ \ell(\hat{f}(X_{n+1}, \hat{\theta}), Y_{n+1}) \mid (X,Y) \right],
$$
where $(X_{n+1}, Y_{n+1}) \sim P$ is an independent test point from the same distribution. Notice $\err_\xy$ is a random quantity, depending on the training data. We denote the expectation of this quantity across possible training sets as
$$
\err := \E\left[\err_\xy\right].
$$
We will discuss the relationship between these two quantities further in Section~\ref{sec:estimand_of_cv}. We note that out-of-sample error is materially different from {\em in-sample-error}
which is the focus of methods like the $C_p$ and AIC statistics, and covariance penalties. These are discussed in Section \ref{subsec:cov_penalties}.

In cross-validation, we partition the observations $\Ical = \{1,\dots,n\}$ into $K$ disjoint subsets (\emph{folds}) $\Ical_1$,\dots, $\Ical_{K}$ of size $m = n / K$ at random. Throughout this work, we will assume $K$ divides $n$ for convenience, and we will choose $K=10$ in all of our numerical results.
Consider the first fold, and let $\hat{\theta}^{(-1)} = \A((X_j, Y_j)_{j \in \Ical \setminus \Ical_1}$ be the model fit to only those points that are not in fold one. Then, let $e_i = \ell(\hat{f}(x_i, \hat{\theta}^{(-1)}), y_i)$ for each $i \in \Ical_1$. The errors $e_i$ for points in other folds are defined analogously. 
We let
\begin{equation}
\errhat^\cv := \bar{e} = \frac{1}{n} \sum_{i=1}^n e_i
\label{eq:cv_errhat}
\end{equation}
be the average error, which is the usual CV estimate of prediction error. If one desires a confidence interval for the prediction error, a straightforward approach is to compute the empirical standard deviation of the $e_i$ divided by $\sqrt{n}$ to get an estimate of the standard error:
\begin{equation*}
\sehat := \frac{1}{\sqrt{n}} \cdot \sqrt{\frac{1}{(n-1)} \sum_{i=1}^n (e_i - \bar{e})^2}.
\end{equation*}
From here, we can create a confidence interval as
\begin{equation*}
    (\bar{e} - z_{1-\alpha/2} \cdot \sehat, \ \bar{e} + z_{1-\alpha/2} \cdot \sehat),
\end{equation*}
where $z_{1-\alpha/2}$ is the $1-\alpha/2$ quantile of the standard normal distribution. We  call these the \emph{na\"ive cross-validation intervals} and they  serve as our baseline approach. Importantly, we find that  these na\"ive CV intervals are on average too small because the true standard deviation of $\bar{e}$ is larger than the na\"ive estimate $\sehat$ would suggest, so a better estimate of the standard error is needed.

As a final piece of notation for our asymptotic statements, for a reference sequence $a_1,a_2,\dots$, $O(a_m)$ denotes a sequence $b_1,b_2,\dots$ such that the sequence $b_1 / a_1, b_2 / a_2, \dots$ has finite limit superior; $\Omega(a_m)$ denotes a sequence $b_1,b_2,\dots$ such that the sequence $b_1 / a_1, b_2 / a_2, \dots$ has positive limit inferior; and $\Theta(a_m)$ denotes a sequence that satisfies both properties. Lastly, $o(a_m)$ denotes a sequence $b_1,b_2,\dots$ such that the sequence $b_1 / a_1, b_2 / a_2, \dots$ converges to zero.

\section{What  prediction error are we estimating?}
\label{sec:estimand_of_cv}
We next discuss targets of inference when assessing prediction accuracy. We discuss both $\err$ and $\err_{XY}$, and also introduce an intermediate quantity $\err_X$ that explains the connection between these two. While cross-validation is our focus, our results hold identically for other estimates of prediction error: covariance penalties (Section~\ref{subsec:cov_penalties}), data splitting (Section~\ref{subsec:data_splitting}), and bootstrap (Appendix~\ref{app:bootstrap_theorems}).

\subsection{$\err_X$: a different target of inference}
\label{subsec:errx_results}

\begin{figure}
    \centering
    \includegraphics[page = 2, width = 4in, trim = 0 32cm 25cm 0, clip]{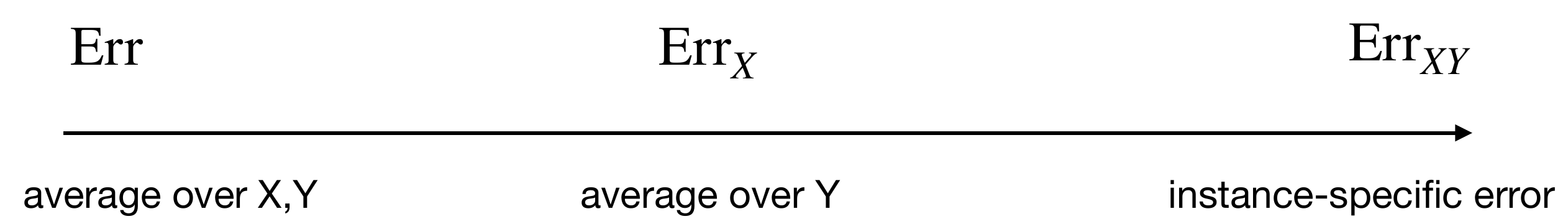}
    \caption{Possible targets of inference for cross-validation. Here, $(X,Y)$ is the training data and $\err_\xy$ is the average error of the model fit on $(X,Y)$ on a test data set of infinite size. From left to right, the random variables above are a constant, a function of $X$ only, and a function of $(X,Y)$.}
    \label{fig:err_tau_viz}
\end{figure}

The two most natural estimands of interest to the analyst are $\err_\xy$, the error of the model that was fit on our actual training set, and $\err$, the average error of the fitting algorithm run on same-sized datasets drawn from the underlying distribution $P$. The former quantity is of the most interest to a practitioner deploying a specific model, whereas the latter may be of interest to a researcher comparing different fitting algorithms. While it may initially appear that the quantity $\err_\xy$ is easier to estimate---since it concerns the model at hand---it has been observed that the cross-validation estimate provides little information about $\err_\xy$ \citep{zhang1995, hastie01statisticallearning, yousef2020leisurely}, a phenomenon sometimes called the \emph{weak correlation} issue.

We now prove that CV has lower MSE for estimating $\err$ than it does for $\err_{XY}$, for the special case of the linear model.  In this sense, CV should be viewed as an estimate of $\err$ rather than of $\err_{XY}$.
In order to state this formally, for this section only, assume the homoskedastic linear model holds:
\begin{equation}
\label{eq:linear_model_assumption}
y_i = x_i^\top \theta + \epsilon_i \quad \text{where} \quad \epsilon_i \simiid \cN(0, \sigma^2) \quad i = 1,\dots,n,
\end{equation}
with $\epsilon = (\epsilon_1,\dots,\epsilon_n)$ independent of $X$.
In this setting, a key quantity in our analysis is
\begin{equation*}
\err_X := \E[\err_\xy \mid X],
\end{equation*}
which falls between $\err$ and $\err_{XY}$; see Figure~\ref{fig:err_tau_viz} for a visualization.
This quantity is also considered by \citet{hastie2019surprises} in a high-dimensional regression setting, but to the best of our knowledge has not been considered in the literature on estimation of prediction error.

While our current focus is on cross-validation, the conclusions hold for a broad class of estimates of prediction error. In particular, we consider estimators of prediction error that satisfy the following property:
\begin{definition}[Linearly invariant estimator]
We say that an estimator of prediction error\\
$\errhat((X_1, Y_1), \dots, (X_n, Y_n), U)$ is \emph{linearly invariant} if for all $x_i, y_i, u$ we have
\begin{equation}
    \errhat\big((x_1, y_1), \dots, (x_n, y_n), u\big) = \errhat\big((x_1, y_1 + x_1^\top \kappa), \dots, (x_n, y_n + x_n^\top \kappa), u\big).
\end{equation}
Here $\kappa$ is any $p$-vector and the random variable $U$ is included to allow for randomized procedures like cross-validation, and without loss of generality it is taken to be $\textnormal{unif}[0,1]$ and independent of $(X,Y)$.
\end{definition}
\noindent With ordinary least square (OLS) fitting, cross-validation satisfies this property:
\begin{lemma}
When using OLS as the fitting algorithm and squared-error loss, the cross-validation estimate of prediction error, $\errhat^\cv$, is linearly invariant.
\label{lem:ols_equivariance}
\end{lemma}
\noindent Note that linear invariance is a deterministic property of an estimator and does not rely on any distributional assumptions.

\begin{figure}[t]
\centering
\begin{subfigure}[b]{.3\textwidth}
\centering
\includegraphics[height = 2in]{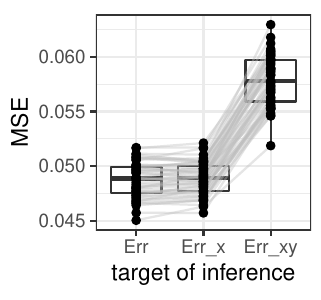}
\end{subfigure}
\hspace{.2cm}
\begin{subfigure}[b]{.3\textwidth}
\centering
\includegraphics[height = 2in]{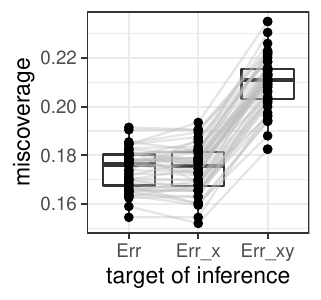}
\end{subfigure}
\hspace{.2cm}
\begin{subfigure}[b]{.3\textwidth}
\centering
\includegraphics[height = 2in]{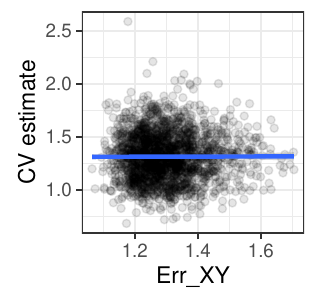}
\end{subfigure}
\hspace{.2cm}
\caption{Left: mean squared error of the CV point estimate of prediction error relative to three different estimands: $\err$, $\err_X$, and $\err_\xy$. Center: coverage of $\err$, $\err_X$, and $\err_\xy$ by the na\"ive cross-validation intervals in a homoskedastic Gaussian linear model. The nominal miscoverage rate is 10\%. Each pair of points connected by a line represents 2000 replicates with the same feature matrix $X$. Right: $2000$ replicates with the same feature matrix and the line of best fit (blue).}
\label{fig:ols_target}
\end{figure}

Recall from classical linear regression theory that when using ordinary least squares (OLS), the estimated coefficient vector is independent of the residual sum of squares. This implies that the sum of squared residuals is independent of the true predictive error. It turns out that even further, the CV estimate of error (and all linearly invariant estimates of error) is independent of the true error, conditional on the feature matrix $X$.
\begin{theorem}
Assume the homoskedastic Gaussian linear model \eqref{eq:linear_model_assumption} holds and that we use squared-error loss. Let $\errhat$ be a linearly invariant estimate of prediction error (such as $\errhat^\cv$ using OLS as the fitting algorithm). Then,
\begin{equation}
\errhat \indep \err_\xy \mid X.
\label{eq:ols_indep_statement}
\end{equation}
\label{thm:ols_cv_independent_err}
\end{theorem}

The proof of this theorem rests primarily on the fact that the OLS residuals are independent of the fitted coefficient vector in the linear model, together with the observation that linearly invariant estimators are a function \emph{only} of the residuals of an OLS model fit. Due to its simplicity, we give the proof explicitly here; all other proofs are given in Appendix~\ref{app:proofs}.
\begin{proof}[Proof of Theorem~\ref{thm:ols_cv_independent_err}]
The true predictive error ($\err_\xy$) is a function only of $\hat{\theta}$, the OLS estimate of $\theta$ based on the full sample $(x_1, y_1), \dots, (x_n, y_n)$. On the other hand, any linearly invariant $\errhat$ is a function only of the residuals $Y - X \hat{\theta} = (I - X (X^\top X)^{-1} X^\top) Y$, by the invariance property (see Lemma~\ref{lem:invariant_ols_residuals}). Since
\begin{equation*}
\hat{\theta} \indep (Y - X \hat{\theta}) \mid X,
\end{equation*}
from classical linear model results, the proof is complete.
\end{proof}

As a result, any linearly invariant estimator (such as cross-validation) has lower MSE as an estimate of $\err_X$ than as an estimate of $\err_{XY}$:
\begin{corollary}
Under the conditions of Theorem~\ref{thm:ols_cv_independent_err},  
\begin{equation*}
    \E\left[\left(\errhat - \err_{XY}\right)^2\right] = \E\left[\left(\errhat - \err_{X}\right)^2\right] + \underbrace{\E\left[\var(\err_{XY}\mid X)\right]}_{\ge 0}.
\end{equation*}
\label{cor:errx_lower_mse}
\end{corollary}

We demonstrate this in an experiment in a simple linear model with $n=100$ observations and $p=20$ features, where the features are i.i.d. standard normal variables; see Figure~\ref{fig:ols_target}. 
As predicted by Corollary~\ref{cor:errx_lower_mse}, we see that the CV point estimate has lower MSE for $\err_X$ than for $\err_{XY}$. Similarly, the na\"ive CV intervals cover $\err_X$ more often than they cover $\err_\xy$. 

\begin{remark} Theorem~\ref{thm:ols_cv_independent_err} can be restated in an evocative way. Suppose we have two data sets $(X, Y)$ and $(X, Y')$ that share the same feature matrix $X$. Let $\err_\xy$ and $\err_{\xy'}$ be the true errors of the model fit with these two data sets, respectively. Next, suppose we perform cross-validation on $(X,Y)$ to get an estimate $\widehat{\err}_\xy$ and do the same on $(X, Y')$ to get an estimate $\widehat{\err}_{\xy'}$. Then, \eqref{eq:ols_indep_statement} is equivalent to
\begin{equation*}
    \big(\err_\xy, \widehat{\err}_\xy \big) \quad \eqd \quad \big(\err_\xy, \widehat{\err}_{\xy'} \big).
\end{equation*}
This means that for the purpose of estimating $\err_\xy$, we have no reason to prefer using the cross-validation estimate with $(X, Y)$ to using the cross-validation estimate with a different data set $(X, Y')$, even though we wish to estimate the error of the model fit on $(X, Y)$. 
\end{remark}

\subsection{Relationship with average error}
\label{subsec:err_errx_results}

The results of the previous section suggest that $\err_{X}$ is a more natural target of inference than $\err_{XY}$. Next, we examine the relationship between $\err$ and $\err_X$, showing that $\err_X$ is close to $\err$, in that the variance of $\err_X$ (which has mean $\err$) is small compared with the variance of $\err_{XY}$ (which also has mean $\err$). Combined with the results of the previous section, this gives a formal statement that cross-validation is a better estimator for $\err$ than for $\err_{XY}$. 

To make this precise, consider the conditional variance decomposition of the variance of $\err_{XY}$,
\begin{equation}
    \var(\err_{XY})
    = \underbrace{\E_{X}\left[\var(\err_{XY} \mid X)\right]}_{\text{var due to $Y \mid X$}} + \underbrace{\var(\err_X)}_{\text{var due to $X$}}.
\label{eq:err_xy_var_decomp}
\end{equation}
To quantify the relative contribution of the two terms in the right-hand side of~\eqref{eq:err_xy_var_decomp}, we will use a \emph{proportional asymptotic limit}, where
\begin{equation}
    n > p, \qquad n,p \to \infty, \qquad n / p \to \lambda > 1.
\label{eq:highd_limit_def}
\end{equation}
We use the proportional asymptotic limit rather than traditional $p$ fixed, $n \to \infty$ asymptotics, because in the latter asymptotic regime, the difference between $\err$, $\err_X$, and $\err_{XY}$ is asymptotic order lower than $1/\sqrt{n}$, so one always estimates these three targets with equal precision, and the analysis is less informative. See Appendix~\ref{app:lowd_asymptotics} for a complementary analysis in the traditional $p$ fixed, $n \to \infty$ asymptotic regime and \citet{yang2007consistency, Wager2019} for a related discussion. By contrast, in the proportional asymptotic limit we will see that $\errhat^\cv$ is closer to $\err$ and $\err_X$ than to $\err_{XY}$.

\begin{theorem}
Suppose the homoskedastic Gaussian linear model in \eqref{eq:linear_model_assumption} holds and that we use squared-error loss. In addition, assume that feature vectors $X_i \sim \mathcal{N}(0, \Sigma_p)$ for any full-rank $\Sigma_p$. Then, in the proportional asymptotic limit in \eqref{eq:highd_limit_def}, we have
\begin{equation*}
    \E_{X}\left[\var(\err_{XY} \mid X)\right] = \Theta(1 / n)
\end{equation*}
and
\begin{equation*}
    \var(\err_X) = \E(\err_X - \err)^2 = \Theta(1 / n^2),
\end{equation*}
as $n,p \to \infty$.
\label{thm:err_and_err_x_higd}
\end{theorem}
We summarize the asymptotic relationship among the various estimands in Figure~\ref{fig:highd_asymptotics_viz}.  We see that the randomness caused by $Y$ given $X$ is of a larger order than that due to the randomness in $X$. This explains why in Figure~\ref{fig:ols_target}, the coverage and MSE of cross-validation is similar when estimating either $\err$ or $\err_X$, but is significantly different when estimating $\err_{XY}$. As a result, $\err_X$ and $\err_{XY}$ are asymptotically uncorrelated, and moreover, combining this with Theorem~\ref{thm:ols_cv_independent_err} shows that $\errhat^\cv$ is asymptotically uncorrelated with $\err_{XY}$, as stated next.

\begin{figure}[!t]
    \centering
    \includegraphics[width = 4in, trim = 0 32cm 38cm 0, clip]{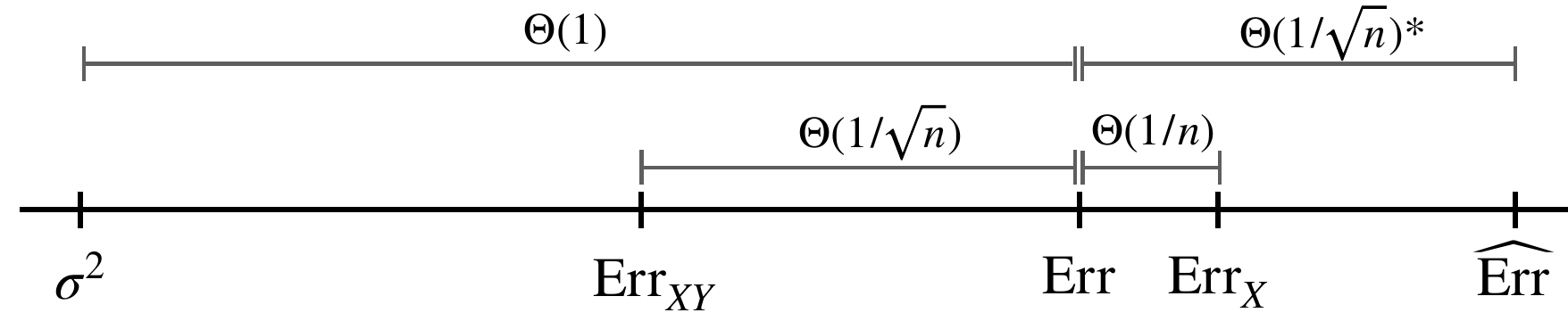}
    \caption{The relationship among various notions of prediction error in the proportional asymptotic limit~\eqref{eq:highd_limit_def}. Recall that $\sigma^2$ is the Bayes error: the error rate of the best possible model. See Figure~\ref{fig:highd_rates} for a simulation experiment demonstrating these rates. $^*$The variance of $\errhat$ scales as $1/\sqrt{n}$; see Section~\ref{subsec:cv_bias} for details about the bias.}
    \label{fig:highd_asymptotics_viz}
\end{figure}

\begin{figure}[!ht]
    \centering
    \includegraphics[height = 2.5in]{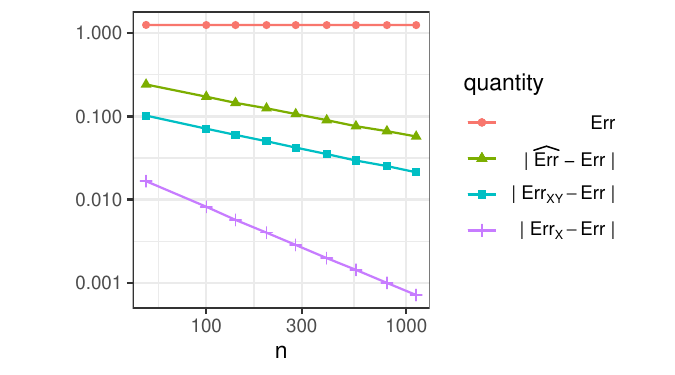}
    \caption{Simulation results demonstrating the asympotic scaling presented in Figure~\ref{fig:highd_asymptotics_viz}. The fitted slopes of the lines (after log-transforming both axes) are $0.00, -0.46, -0.50, -1.01$, from top to bottom. See Section~\ref{subsec:cv_bias} for details about the rate of $\errhat$}
    \label{fig:highd_rates}
\end{figure}

\begin{corollary}
In the setting of Theorem~\ref{thm:err_and_err_x_higd},
\begin{equation*}
    \cor(\err_{XY}, \err_X) \to 0 \quad \text{as } n,p \to \infty.
\end{equation*}
Moreover, for any linearly invariant estimator $\errhat$ (such as $\errhat^\cv$ using OLS as the fitting algorithm),
\begin{equation*}
    \cor(\err_{XY}, \errhat) \to 0 \quad \text{as } n,p \to \infty.
\end{equation*}
\label{cor:cor_errx_errxy_highd}
\end{corollary}
Notice that this is a marginal result, whereas the similar Theorem~\ref{thm:ols_cv_independent_err} is conditional on $X$. With respect to Figure~\ref{fig:highd_asymptotics_viz}, this result means that the fluctuations of $\err_{XY}$ around $\err$ are asymptotically uncorrelated with the fluctuations of of $\errhat$ around $\err$.
Combining Theorem~\ref{thm:err_and_err_x_higd} with Theorem~\ref{thm:ols_cv_independent_err}, we conclude that CV has larger error for estimating $\err_{XY}$ than for $\err$ or $\err_X$:
\begin{corollary}
In the setting of Theorem~\ref{thm:err_and_err_x_higd}, let $\errhat$ be any linearly invariant estimator (such as $\errhat^\cv$ using OLS as the fitting algorithm). Suppose in addition that $\var(\errhat) \to 0$ (an extremely weak condition satisfied by any reasonable estimator). Then,
\begin{align*}
    \E\left[\left(\errhat - \err_{XY}\right)^2\right] - \E\left[\left(\errhat - \err_X \right)^2\right] &= \Omega(1 / n), \\ 
      \E\left[\left(\errhat - \err_{XY}\right)^2\right] - \E\left[\left(\errhat - \err \right)^2\right] &= \Omega(1 / n), \quad \text{ and } \\ 
        \left\lvert \E\left[\left(\errhat - \err\right)^2\right] - \E\left[\left(\errhat - \err_X \right)^2\right] \right\rvert &= o(1 / n). 
\end{align*}
\label{cor:err_errxy_gap_highd}
\end{corollary}

The asymptotic theory perfectly predicts the experimental results presented in Figure~\ref{fig:highd_rates}; we see that even for moderate sample size, the scalings is exactly as anticipated. The main conclusion is that for a linearly invariant estimate of prediction error that has precision $1 / \sqrt{n}$, our results show that asymptotically one has lower estimation error when estimating $\err$ compared to $\err_{XY}$. Similarly, the correlation between a linearly invariant estimate and $\err_{XY}$ goes to zero. These theoretical predictions are also corroborated by the experimental results presented in Figure~\ref{fig:highd_rates2}. Thus, cross-validation is estimating the average error $\err$ more so than the specific error $\err_{XY}$. 

\begin{remark}
Note that the results in this section apply both to $K$-fold cross-validation with fixed $K$, and leave-one-out cross validation where $K = n$. Formally, the results require only that one is using some sequence of linearly invariant estimators.
\end{remark}

\subsection{The bias of cross-validation}
\label{subsec:cv_bias}

Up until this point, we have not explicitly mentioned the bias in the CV point estimate $\errhat$ that comes from the difference in sample size. That is, $\errhat$ uses models of size $n (K - 1) / K$, whereas $\err$ and $\err_{XY}$ are defined for models fit on data of size $n$, so $\E[\errhat]$ is typically smaller than $\err = \E[\err_{XY}]$. We now pause for a few remarks about this bias. First, notice that Corollary~\ref{cor:err_errxy_gap_highd} sidesteps the bias issue by considering differences between two mean squared error quantities. The bias is important, however, if we wish to understand absolute quantities such as $\E[(\errhat - \err)^2]$. To this end, the bias exhibits different behavior in different regimes\footnote{We thank an anonymous reviewer for feedback on this topic.}:
\begin{itemize}
    \item {\em The parametric regime.} Suppose $p$ is fixed, $n \to \infty$, and the model class has fixed dimension. Here, the bias will typically be of order $1/n$, which means that it is negligible compared to the variance. (In fact, the dimension of the model class can grow, provided the rate is slow enough; see~\citet{Wager2019} for discussion.)
    \item {\em The proportional, dense regime.} Consider the setting above in~\eqref{eq:highd_limit_def}, fitting a dense model. If the number of folds is fixed, the bias of $\errhat$ will converge to a nonzero constant as $n$ and $p$ grow~\citep[e.g.,][]{Liu2020Ridge}. What this means is that in Figure~\ref{fig:highd_rates}, the $|\errhat - \err|$ curve will eventually cease to decay at a $1/\sqrt{n}$ rate, bottoming out due to the constant bias. We do not see this behavior in the plot, because the bias is still much smaller than the variance at the sample sizes we consider; Figure~\ref{fig:ols_rates_bv} reports the bias and variance in this setting.
    \item {\em The proportional, sparse regime.} The setting is the most delicate. Here, the behavior of sparse regression algorithms may have very different behavior on samples of size $n (K - 1) / K$ versus samples of size $n$~\cite[e.g.,][]{reeves19a}. Thus, the bias here may be appreciable.
\end{itemize}
In all cases, the bias can be mitigated by taking a larger number of folds as $n$ grows.

Incorporating bias alongside our results from Section~\ref{subsec:err_errx_results} leads to an interesting bias-variance-variance decomposition of $\E[(\errhat - \err_{XY})^2]$. Since both $\errhat$ and $\err_{XY}$ are random quantities, we cannot use the usual bias-variance decomposition. However,
by Corollary~\ref{cor:cor_errx_errxy_highd}, these two quantities are asymptotically uncorrelated, yielding the following:
\begin{equation*}
    \E\bigg[(\errhat - \err_{XY})^2\bigg] \approx 
    \underbrace{\bigg(\E\left[\errhat\right] - \err\bigg)^2}_{\text{bias}} + 
    \underbrace{\E\bigg[(\errhat - \E[\errhat])^2\bigg]}_{\text{variance of $\errhat$}} + 
    \underbrace{\E\bigg[(\err_{XY} - \err)^2 \bigg]}_{\text{variance of $\err_{XY}$}}.
\end{equation*}
The first two terms on the right hand side are the bias-variance decomposition for $\errhat$ as an estimate of $\err$. Thus, because the additional third term is positive, we again see that $\errhat$ is a more precise estimate of $\err$ than of $\err_{XY}$. 


\subsection{Data splitting}
\label{subsec:data_splitting}

Perhaps the simplest way to estimate prediction error is to split the data into two disjoint sets, one for training and one for estimating the prediction accuracy. The previous results also shed light on the properties of data splitting. In particular, we will show that when estimating prediction error with data splitting, refitting the model on the full data incurs additional variance that can make the confidence intervals slightly too small, even asymptotically. This is not a cause for practical concern, but it is another manifestation of the fact that linearly invariant estimators are estimating average prediction error, and $\err_{XY}$ contains additional, independent variation. We report on the details in Appendix~\ref{app:data_splitting}.


\subsection{Connection with covariance penalties}
\label{subsec:cov_penalties}

For parametric models, there is an alternative theory of the estimation of prediction accuracy based on \emph{covariance penalties}; see \citet{stein1981estimating, efron2004estimation, Rosset2020} for overviews of this approach. For the linear model with OLS and squared error loss, this approach specializes to the well-known Mallows $C_p$ \citep{mallows1973comments, akaike1974aic} estimate of prediction error:
\begin{equation*}
\errhat^\cp := \frac{1}{n} \sum_{i=1}^n (y_i - \hat{f}(x_i, \hat{\theta}))^2 + \frac{2 p \hat{\sigma}^2}{n}.
\end{equation*}
The classical covariance penalty approach is focused on estimating \emph{in-sample error}, the error for a fresh sample with the same features $X$:
\begin{equation*}
\err_{\inn}(X) := \E\left[\frac{1}{n}\sum_{i=1}^n (Y_i' - \hat{f}(X_i, \hat{\theta}))^2 \mid X\right],
\end{equation*}
where the expectation is only over $Y_i$ and $Y_i'$ for $i=1,\dots,n$, and $Y_i, Y_i'$ are independent draws from the distribution of $Y_i \mid X_i = x_i$. See Figure~\ref{fig:error_notions} for a visualization of how this relates our other notions of prediction error. 
To extend this to the setting where the feature vector of future test points is also random (the setting of the present work), \citet{Rosset2020} introduce $RC_p$, which is a similar but slightly larger estimate of prediction error that accounts for the variability in the features of future test points.

\subsubsection{Estimand of $C_p$}
We first discuss the estimation of prediction error with Mallows $C_p$, showing that (like cross-validation) it is a worse estimator for $\err_{XY}$ than for $\err$. The results from Section~\ref{subsec:errx_results} and Section~\ref{subsec:err_errx_results} continue to hold for $\errhat^\cp$ and $\errhat^\rcp$, since they are linearly invariant:
\begin{lemma}
The estimators $\errhat^\cp$ and $\errhat^\rcp$ are linearly invariant.
\end{lemma}
This result is immediate from the fact the $\errhat^\cp$ and $\errhat^\rcp$ are functions only of the residuals of the OLS fit. Thus, the conclusions of Theorem~\ref{thm:ols_cv_independent_err}, Corollary~\ref{cor:errx_lower_mse}, Corollary~\ref{cor:cor_errx_errxy_highd}, and Corollary~\ref{cor:err_errxy_gap_highd} hold for $\errhat^\cp$. In particular, $\errhat^\cp$ has lower error for estimating $\err$ and  $\err_X$ than for estimating $\err_{XY}$, and $\errhat^\cp$ is asymptotically uncorrelated with $\err_{XY}$. In summary, as before with cross-validation, Mallow's $C_p$ is not able to estimate $\err_{XY}$, but is rather an estimate of $\err_\inn$, $\err$ or $\err_X$ (the latter two are close for large samples).

\subsubsection{A decomposition of $\err_X$}

Next, we develop a decomposition  for $\err_X$. The results that we present are implicit in \citet{Rosset2020}, but in that work the results are averaged over $X$ to instead obtain estimates for $\err$. From the definitions of $\err_X$ and $\err_\inn$, we trivially have that
\begin{equation}
   \mbox{Err}_X= \mbox{Err}_{\rm in}(X) + \E \left[ \E \bigg[(Y_{n+1} - \hat{f}(X_{n+1}, \hat{\theta}))^2 \mid (X, Y) \bigg] - \E\bigg[\frac{1}{n}\sum_{i=1}^n (Y_i' - \hat{f}(X_i, \hat{\theta}))^2 \mid (X, Y) \bigg] \mid X\right].
   \label{eq:decomp}
 \end{equation}
 
 When the linear model holds, this simplifies as stated next.
\begin{proposition}
For the linear model with the OLS fitting algorithm and squared error loss, assume in addition that the distribution of $X_i$ has mean zero and covariance $\Sigma$ of full rank. Then,
\begin{equation*}
     \err_X= \err_\inn(X) +
     \frac{\sigma^2}n\left(\tr(\widehat\Sigma^{-1}\Sigma) - p\right),
\end{equation*}
where $\widehat{\Sigma} = X^\top X / n$.
\label{prop:errx_cov_penalty}
\end{proposition}
The second term in the sum can be either positive or negative. Roughly speaking, this term is smaller (more negative) if $X$ is a good design that yields precise estimates of the regression coefficients, whereas this term is larger (more positive) if $X$ yields less precise estimates. Note that we do not typically know the population covariance $\Sigma$, so this cannot be used as an estimator for prediction error. Instead, it is an expository decomposition relating $\err_X$ with existing work about the estimation of prediction error.
From this expression, we can read off the following results from \citep{Rosset2020}.

\begin{corollary}[Random-X covariance penalties \citep{Rosset2020}]
In the setting of Proposition~\ref{prop:errx_cov_penalty}, we have that 
\begin{equation*}
\err \ge \E[\err_\inn].
\end{equation*}
Moreover, if the feature vector follows a multivariate Gaussian distribution, then 
\begin{equation*}
\E[\err_X] = \E[\err_\inn] + \frac{p\sigma^2}{n}\left(\frac{p + 1}{n-p-1}\right).
\end{equation*}
\label{cor:randomx_cov_penalty}
\end{corollary}
The latter expression is the motivation for the $RC_p$ penalty.

\begin{figure}
    \centering
    \includegraphics[height = 1in, trim = 0 700 700 0, clip]{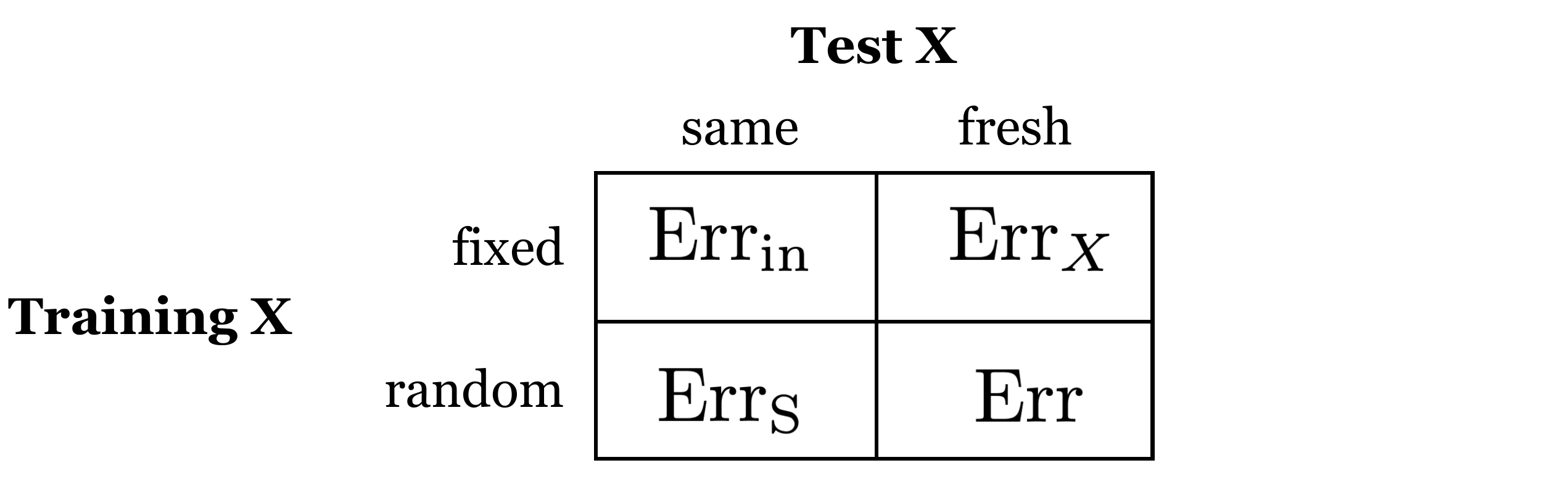}
    \caption{A representation of the notions of error considered in this work. See \citet{Rosset2020} for a definition and discussion of $\err_{\textnormal{S}}$.}
    \label{fig:error_notions}
\end{figure}
\subsection{Bootstrap estimates of prediction error}
Bootstrap estimates of prediction error are also linearly invariant, and so they are also estimates of the average prediction error. For brevity, we present these results in Appendix~\ref{app:bootstrap_theorems}.

\section{Confidence intervals with nested cross-validation}

In this section, we develop an estimator for the MSE of the cross-validation point estimate. Our ultimate goal is then to use the estimated MSE to give confidence intervals for prediction error with approximately valid coverage.

\subsection{Dependence structure of CV errors}
\label{sec:CV_cov_structure}
Before developing our estimator for the cross-validation MSE, we pause here to build up intuition for why the na\"ive CV confidence intervals for prediction error can fail, as seen previously in our example in Section~\ref{subsec:harrell_model}.
The na\"ive CV intervals are too small, on average, because the true sampling variance of $\errhat^\cv$ is larger than the na\"ive estimate $\sehat$ would suggest. In particular, this estimate of the variance of the  CV point estimate assumes that the observed errors $e_1,\dots,e_n$ are independent. This is not true---the observed errors have less information than an independent sample since each point is used for both training and testing, which induces dependence among these terms. 

\begin{figure}[t]
    \centering
    \includegraphics[height = 1.75in]{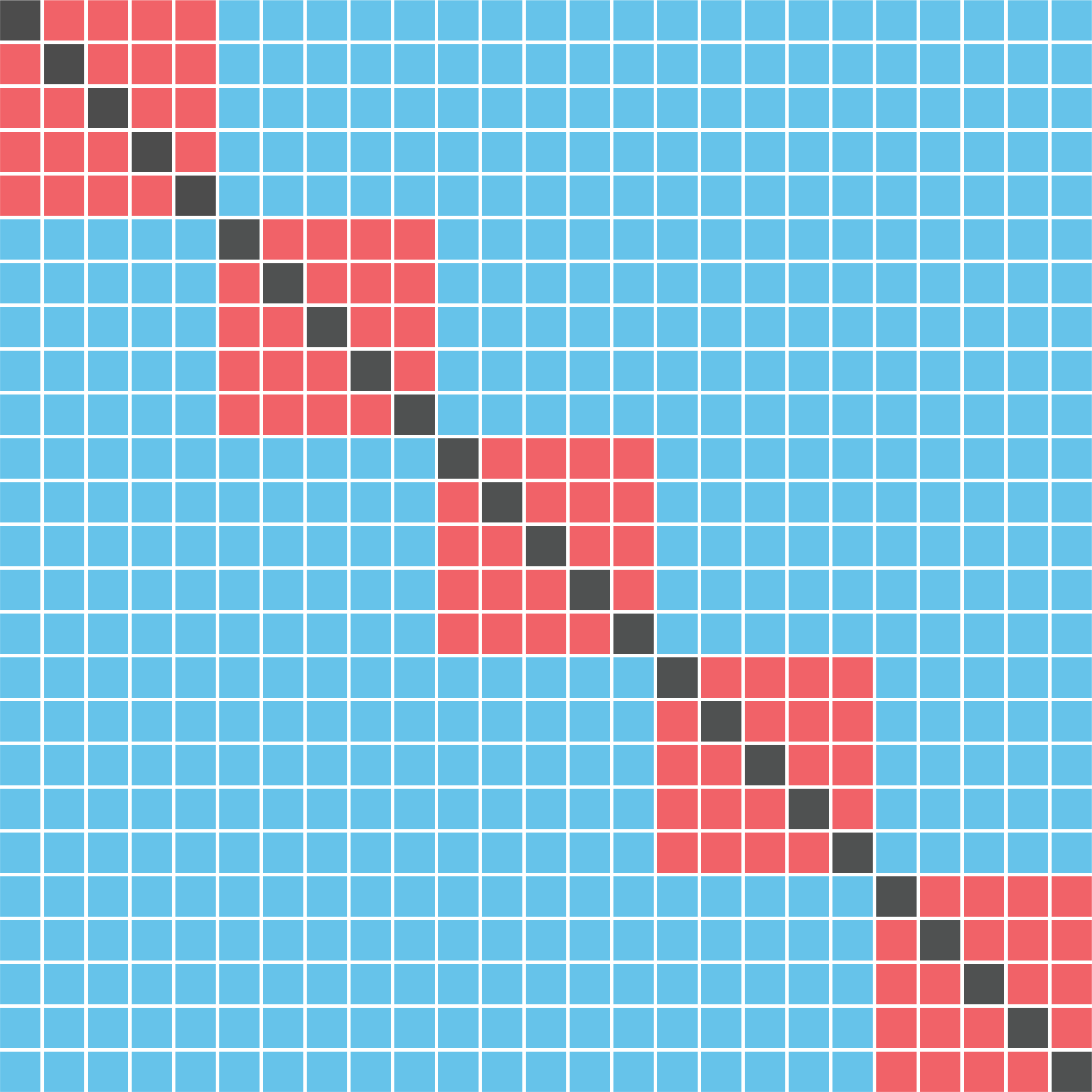}
    \caption{Covariance structure of the CV errors. Red entries correspond to the covariance between points in the same fold, and blue entries correspond to the covariance between points in different folds.}
    \label{fig:covariance_heatmap}
\end{figure}

In particular, the covariance matrix of the errors $e_1,\dots,e_n$ has the block structure shown in Figure~\ref{fig:covariance_heatmap}. As such, it is parameterized by only three numbers: $a_1 := \var(e_1)$, $a_2:= \cov(e_i, e_j)$  for $i,j$ in the same fold, and $a_3 := \cov(e_i, e_j)$ for $i,j$ in different folds \citep{bengio2004no}.
It is easy to see that the variance of $\bar{e}$ is
\begin{equation*}
\var(\bar{e}) = \frac{1}{n} a_1 + \frac{n/K-1}{n} a_2 + \frac{n-n/K}{n} a_3;
\end{equation*}
see \citet{bengio2004no}.
The constants $a_2$ and $a_3$ will typically be positive, in which case $\var(\bar{e}) > a_1 / n$, and so estimating the variance of $\bar{e}$ as $\sehat^2$ results in an estimate that is too small.
For example, in the setting from Section~\ref{subsec:harrell_model}, the estimated variance is approximately a factor of $2.65$ too small, so the na\"ive confidence intervals are too small by a factor of $\sqrt{2.65} \approx 1.6$. In order to correctly estimate the variance of $\bar{e}$, it would suffice to estimate $a_1$, $a_2$, and $a_3$, but
\citet{bengio2004no} proves surprising fact that there is no unbiased estimator of $\Var(\bar{e})$ based on a single run of cross-validation. Thus, estimating $a_1$, $a_2$, and $a_3$ cannot be done from a single run of cross-validation. Although this result suggests that something beyond the usual cross-validation will be required to give good estimates of the standard error of $\errhat^\cv$, it does not imply that it is impossible to get such estimates with other approaches.

To recap, the primary issue with the na\"ive cross-validation confidence intervals is that they rely on an independence assumption that is violated: they implicitly assume that $a_2 = 0$ and $a_3 = 0$. Thus, the usual estimate of the variance of $\errhat^\cv$ is too small, resulting in poor coverage. To remedy this issue, we  develop an estimator that empirically estimates  the variance of $\errhat^\cv$ across many subsamples. Avoiding the faulty independence approximation leads to intervals with superior coverage. We turn to details of our proposed procedure next.

\subsection{Our target of inference}
\label{subsec:ncv_target}

In this section, our primary goal will be to give confidence intervals for test accuracy by estimating the mean-squared error (MSE) of cross-validation:
\begin{definition}\textbf{}
For a sample of size $n$ split into $K$ folds, the \emph{cross-validation MSE} is
\begin{equation}
\label{eq:mse_def}
\mse_{K,n} := \E\left[\left(\widehat{\err}^\cv - \err_{XY}\right)^2\right].
\end{equation}
\end{definition}
In particular, we define MSE with respect to $\err_{XY}$ and thus will calibrate our test intervals to cover the quantity $\err_{XY}$. At this point, the reader should wonder why we define the MSE with respect to $\err_{XY}$ in view of the results from Section~\ref{sec:estimand_of_cv} that show that $\errhat^\cv$ is a more precise estimate of $\err$ than of $\err_{XY}$, and we will next discuss this issue carefully.  

To be clear, we choose to pursue confidence intervals for $\err_{XY}$ because we are able to do so; the MSE quantity above can be estimated in a convenient way due to an upcoming decomposition (Lemma~\ref{lem:holdout_MSE_identity}). At present, we do not know how to obtain a similar MSE estimate with respect to $\err$. Second, we emphasize that our results from Section~\ref{sec:estimand_of_cv} do \emph{not} mean that giving confidence intervals for $\err_{XY}$ is impossible. Rather, our results say that in the linear model, \emph{confidence intervals for $\err_{XY}$ will be larger than confidence intervals for $\err$}. Still, confidence intervals for either $\err$ or $\err_{XY}$ would be of interest to the analyst. We are able to derive an estimator for the MSE with respect to $\err_{XY}$, and we will turn to the details next.


The MSE in~\eqref{eq:mse_def} contains both a bias term (due to the reduced sample size used by $\errhat^\cv$) and variance term. 
See Section~\ref{subsec:cv_bias} for a discussion of the bias.
Thus, we can view the MSE as a slightly conservative version of the variance of the cross-validation estimator. In any case, the MSE is the relevant quantity for creating confidence intervals around a possibly biased point estimate, since it accounts for both bias and variance. With this in mind, we will use an estimate of the MSE to construct confidence intervals for $\err_{XY}$.
Previewing the remainder of this section, we will produce confidence intervals for $\err_{XY}$ as follows:
\begin{equation}
\left(\errhat^\ncv - \widehat{\bias} - z_{1 - \alpha / 2} \cdot \sqrt{\msehat}, \quad \errhat^\ncv - \widehat{\bias} + z_{1 - \alpha / 2} \cdot \sqrt{\msehat}\right).
\label{eq:finalConfInt}
\end{equation}
Above, $\errhat^\ncv$ is similar to the CV estimate of error except across many random splits, and $\msehat$ is our estimator for MSE---the heart of this section.
In addition, we allow for the possibility of correcting for the sample size bias with an estimator $\widehat{\bias}$; we use one that arises naturally from the computations already carried out to estimate the MSE (Section~\ref{sec:bias_estimate}). 

\subsection{A nested CV estimate of MSE}

\subsubsection{A holdout MSE identity}

We now give a generic decomposition of the mean-squared error of an estimate of prediction error, which will enable use to estimate $\mse_{K,n}$. Consider a single split of the data into a training set and holdout set, i.e., we partition $\Ical=\{1,\dots,n\}$ into $\Ical_\train$ and $\Ical_\test$ calling the training set $(\widetilde{X}, \widetilde{Y})$. Using only $(\widetilde{X}, \widetilde{Y})$, we use our fitting procedure to obtain estimated parameters $\hat{\theta}^\train = \mathcal{A}(\widetilde{X}, \widetilde{Y})$, and further assume we have some estimate of $\widehat{\err}_{\widetilde{X}\widetilde{Y}}$ of the prediction error $\err_{\widetilde{X}\widetilde{Y}}$ defined in \eqref{eq:err_xytilde_def}. Here, $\widehat{\err}_{\widetilde{X}\widetilde{Y}}$ is any estimator of $\err_{\widetilde{X}\widetilde{Y}}$ based only on $(\widetilde{X}, \widetilde{Y})$, such as cross-validation using only $(\widetilde{X}, \widetilde{Y})$.
Let $\{e_i^\out\}_{i \in \Ical_\test}$ be the losses of the fitted model $\hat{f}(\cdot,\hat{\theta}^\train)$ on the holdout set, and let $\bar{e}^\out$ be their average.
The MSE of $\widehat{\err}_{\widetilde{X}\widetilde{Y}}$ can be written as follows:
\begin{lemma}[Holdout MSE identity]
In the setting above
\begin{equation}
\label{eq:holdout_se_identity}
    \underbrace{\E\left[\left(\widehat{\err}_{\widetilde{X}\widetilde{Y}} - \err_{\widetilde{X}\widetilde{Y}}\right)^2\right]}_{\mse} = 
    \underbrace{\E\left[\left(\widehat{\err}_{\widetilde{X}\widetilde{Y}} - \bar{e}^\out\right)^2\right]}_{\textnormal{(a)}} -  \underbrace{\E\left[\left(  \bar{e}^\out-\err_{\widetilde{X}\widetilde{Y}}  \right)^2\right]}_{\textnormal{(b)}}.
\end{equation}
\label{lem:holdout_MSE_identity}
\end{lemma}
The expectations above are over the complete data $(X,Y)$. The lemma follows from adding and subtracting $\err_{\widetilde{X}\widetilde{Y}}$ within term (a) then showing the cross-term is zero with a nested conditional expectation argument. 

This identity is of interest, since both (a) and (b) can be estimated from the data, which leads to an estimate of the MSE term. Specifically,  we  propose the  following estimation strategy:

\medskip

\begin{enumerate}
\item Repeatedly split the data into $\Ical_\train$ and $\Ical_\out$, and for each split, do the following: 
\begin{enumerate}[(i)]
    \item Apply cross-validation to $\Ical_\train$ to obtain $\widehat{\err}_{\widetilde{X}\widetilde{Y}}$ and use $\Ical_\out$ to obtain $\bar{e}^\out$, and then estimate (a) with $(\widehat{\err}_{\widetilde{X}\widetilde{Y}}$ - $\bar{e}^\out)^2$.
    \item Estimate (b) with empirical variance of $\{e_i\}_{i \in \Ical_\out}$ divided by the size of $\Ical_\out$.
\end{enumerate}
 \item  Average together estimates of $(a)$ and $(b)$ across all random splits and take their difference as in \eqref{eq:holdout_se_identity} to obtain an estimate of MSE.
 \end{enumerate}
 
 \medskip
 
Note that the estimates for both (a) and (b) are unbiased, so the resulting MSE estimate is unbiased for the MSE term in \eqref{eq:holdout_se_identity}. In the next section, we will pursue this strategy for the particular case where $\errhat_{\widetilde{X}\widetilde{Y}}$ is itself a cross-validation estimate based only on $(\widetilde{X}, \widetilde{Y})$.

\subsubsection{The MSE estimator}

Building from Lemma~\ref{lem:holdout_MSE_identity}, we now turn to our proposed estimate of MSE, the heart of this section. We follow the estimation strategy described above, using 
 $(K-1)$-fold CV as the estimator $\errhat_{\widetilde{X}\widetilde{Y}}$. This gives an estimate of (a) and (b), and hence an estimate for the MSE, as described above. We also get a point estimate of error by taking the empirical mean of $\widehat{\err}_{\widetilde{X}\widetilde{Y}}$ across the many splits. See Figure~\ref{fig:nested_viz} for a visualization of the nested CV sample splitting, and see Algorithm~\ref{alg:nested_cv} for a detailed description. We denote the resulting estimate of mean squared error by $\msehat^\ncv$ and  the point estimate for prediction error $\errhat^\ncv$.

\begin{figure}
    \centering
    \includegraphics[width = 5in, trim = 0 30cm 37cm 0, clip]{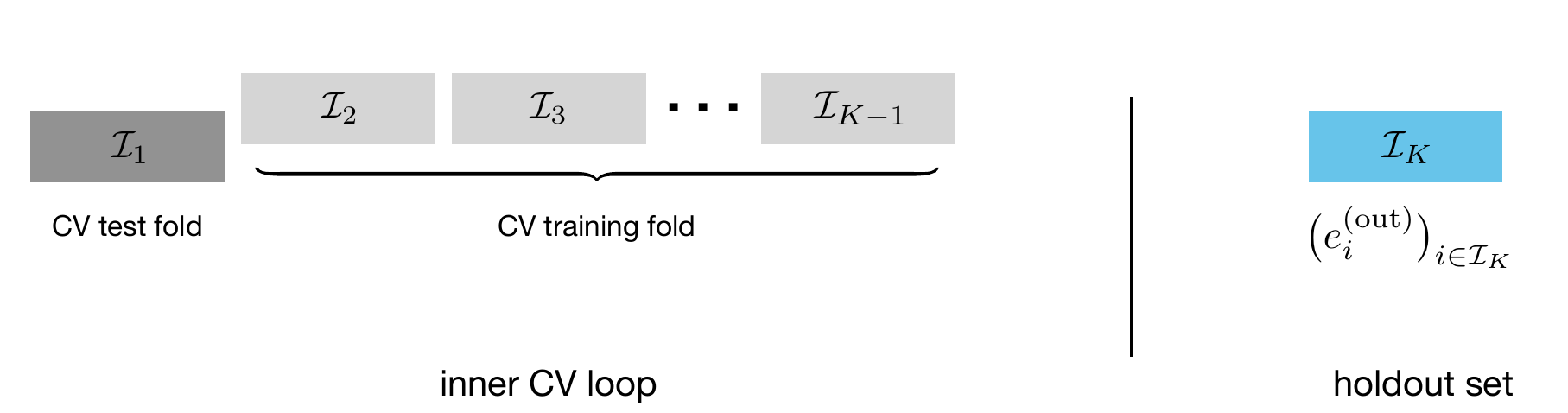}
    \caption{Visualization of nested CV. Using only the folds on left of the vertical line, we perform the usual cross-validation by holding out one fold at a time (the dark grey fold) and then fitting on the remaining folds (the light grey folds). The fresh holdout points (the blue fold) are never used in the inner CV step. }
    \label{fig:nested_viz}
\end{figure}

\begin{algorithm}[ht]
\begin{algorithmic}
\Require data $(X,Y)$, fitting algorithm $\A$, loss $\ell$, number of folds, $K$, number of repetitions $R$

\vspace{.15cm}

\Procedure{nested\_crossval($X, Y$)}{} \Comment{primary algorithm}

\State $\texttt{es} \gets []$ \Comment{initialize empty vectors}
\State $\texttt{a\_list} \gets []$ \Comment{(a) terms} 
\State $\texttt{b\_list} \gets []$ \Comment{(b) terms}

\For{$r \in \{1,\dots,R\}$}
\State Randomly assign points to folds $\Ical_1,\dots,\Ical_K$
\For{$k \in \{1,\dots,K\}$} \Comment{outer CV loop}
	\State $e^{(\inn)} \gets$ \textsc{inner\_crossval}($X, Y, \{\Ical_1,\dots,\Ical_K\} \setminus \Ical_k$) \Comment{inner CV loop}
	\State $\hat{\theta} \gets \A\left((X_i, Y_i)_{i \in \Ical\setminus \Ical_k}\right)$
	\State $e^\out \gets \left(\ell(\hat{f}(X_i, \hat{\theta}), Y_i)\right)_{i \in \Ical_k}$

	\State $\texttt{a\_list} \gets \text{append}\left(\texttt{a\_list}, \left(\textnormal{mean}(e^{(\inn)}) - \textnormal{mean}(e^\out)\right)^2 \right)$
	\State $\texttt{b\_list} \gets \text{append}\left(\texttt{b\_list},\textnormal{var}(e^\out) / |\Ical_k|\right)$
	\State $\texttt{es} \gets \text{append}(\texttt{es}, e^{(\inn)})$
\EndFor
\EndFor

\State $\msehat \gets \text{mean}(\texttt{a\_list}) - \text{mean}(\texttt{b\_list})$ 
\Comment{plug-in estimator based on \eqref{eq:holdout_se_identity}}

\State $\errhat^\ncv \gets \textnormal{mean}(\texttt{es})$

\State \textbf{return:} $(\errhat^\ncv, \msehat)$ \Comment{prediction error estimate and MSE estimate}

\EndProcedure

\vspace{.15cm}

\Procedure{inner\_crossval($X, Y, \{\Ical_1,\dots, \Ical_{K-1}\}$)}{} \Comment{inner cross-validation subroutine}
\State $e^{(\inn)} \gets []$
\For{$k \in \{1,\dots,K-1\}$}
	\State $\hat{\theta} \gets \A\left((X_i, Y_i)_{i \in \Ical_1\cup\dots\cup\Ical_{K-1}\setminus \Ical_k}\right)$
	\State $e^{(\textnormal{temp})} \gets \left(\ell(\hat{f}(X_i, \hat{\theta}), Y_i)\right)_{i \in \Ical_k}$ 
	\State $e^{(\inn)} \gets \text{append}(e^{(\inn)}, e^{(\textnormal{temp})})$
\EndFor
\State \textbf{return:} $e^{(\inn)}$

\vspace{.15cm}

\Ensure $\textsc{nested\_crossval}(X, Y)$  
\EndProcedure
\end{algorithmic}

\caption{Nested Cross-validation}
\label{alg:nested_cv}
\end{algorithm}

In view of Lemme~\ref{lem:holdout_MSE_identity}, we see that the estimator $\msehat^\ncv$ is targeting the MSE of $\errhat^\cv$ as an estimate of $\err_{XY}$, as we record formally next.
\begin{theorem}[Estimand of nested CV]
\label{thm:ncv_unbiased}
For a nested CV with a sample of size $n$,
\begin{equation*}
    \E\left[\msehat^\ncv\right] = \mse_{K-1,n'},
\end{equation*}
where $n' = n (K-1) / K$.
\end{theorem}
This result shows that $\msehat^\ncv$ obtained by nested  CV is estimating the MSE of $(K-1)$-fold cross-validation on a sample of size $n (K-1)/K$. 
Since nested CV uses an inner loop with samples of size $n (K-1) / K$, we recommend re-scaling to obtain an estimate for a sample of size $n$ by instead taking $\msehat = (K-1) / K \cdot \msehat^\ncv$ (although this re-scaled version is not guaranteed to be exactly unbiased for $\mse_{K,n}$). As a minor detail, in practice we also restrict $\sqrt{\msehat}$ to fall between $\sehat$ (the estimated standard error if one had $n$ independent points) and $\sqrt{K} \cdot \sehat$ (the estimated standard error if one had only $n/K$ independent points). This is a minor implementation detail prevents implausible values of $\msehat$ from arising. After adjusting the point estimate $\errhat^\ncv$ with a bias correction discussed next, we form our final confidence intervals as in \eqref{eq:finalConfInt}.

\begin{remark}[An interpretation of nested CV]
A simpler, perhaps more natural approach to this problem would be to borrow ideas from bootstrap calibration \cite[e.g.,][p.~263]{efron1993introduction}. Specifically, one could use a nested CV
scheme, and compute the standard normal confidence interval for error from the inner folds. Then one could check how well this interval covers $\bar{e}^\out$, and adjust the interval to achieve the desired coverage. The problem with this approach is that the left-out fold is finite, so that the interval is a prediction interval for $\bar{e}^\out$ rather than a confidence interval for the true underling prediction error. 

NCV is similar in spirit to this:  over repeated subsamples, we estimate the quantity (a) in \eqref{eq:holdout_se_identity}, which leads to an empirical estimate of how much an interval around $\errhat_{\widetilde{X}\widetilde{Y}}$ must be widened in order to cover $\bar{e}^\out$. The latter is a random quantity, however, so this should be thought of as a calibrated prediction interval. In truth, we wish to cover not $\bar{e}^\out$, but its mean $\err_{\widetilde{X}\widetilde{Y}}$. We convert from a prediction interval for $\bar{e}^\out$ to a confidence interval for $\err_{\widetilde{X}\widetilde{Y}}$ by subtracting out the term (b).
\end{remark}

\begin{remark}[The sample size difference and the target of inference.]
Note that the estimator $\errhat^{\ncv}$ uses models fit on with $n (K - 2) / K$ data points, whereas the target of inference in Theorem~\ref{thm:ncv_unbiased}, $\mse_{K-1,n'}$, is defined with respect to the prediction accuracy of a model fit with $n (K - 1) / K$ points. How can the former be used to estimate the latter? The answer is that the nested CV procedure relies also on fits of size $n (K -1 ) / K$, see the definition of $\hat{\theta}$ in the ``nested\_crossval'' subroutine of Algorithm~\ref{alg:nested_cv}. Nested CV compares the predicted accuracy on the models from $n (K - 2) / K$ data points to the estimated accuracy of the models with $n (K - 1) / K$ data points, using the extra holdout data to asses this accuracy.
\end{remark}

\subsubsection{Estimation of bias}
\label{sec:bias_estimate}

The nested CV computations also yield a convenient estimate of the bias of the NCV point estimate of error, $\errhat^\ncv$. They key idea is that nested CV considers both models fit with $n (K - 2) / K$ data points and with $n (K - 1) / K$ data points, and comparing their these models gives an estimate of bias; see Appendix~\ref{app:bias_estimation}. This aspect of nested CV is not critical---the MSE estimation above is the core of our proposal.

\section{Simulation experiments}
\label{sec:experiments}
We now explore the coverage of nested CV in a variety of settings. In each case, we will report the coverage of na\"ive CV (CV), nested CV (NCV), and data splitting with refitting (DS), where the nominal miscoverage rate is 10\% (5\% miscoverage in each tail).
We also report on the width of the intervals, expressed relative to the width of the standard CV intervals. (We wish to produce intervals that are as narrow as possible while maintaining correct coverage.) We use 10-fold CV (the number of folds has little impact; see Appendix~\ref{subapp:num_folds_exp}) and NCV, with 200 random splits for the latter; see Appendix~\ref{app:compute_times} for the runtime of each experiment. For classification examples we use binary loss, and form confidence intervals for CV, NCV, and data splitting after taking the binomial  variance-stabilizing transformation, described in detail in Appendix~\ref{app:variance_stabilizing_transform}. For regression examples, we use squared error loss.

For data splitting, we use 80\% of the samples for training and 20\% for estimating prediction error. Note that the data splitting without refitting intervals are the same as the data splitting with refitting intervals; the difference is that they are intended to cover different quantities. To make this comparable to CV and nested CV, we report on the coverage of $\err$ and $\err_{XY}$ here, which corresponds to data splitting with refitting. Data splitting without refitting (which seeks to cover the quantity in \eqref{eq:err_xytilde_def}) will typically have better coverage; we observed relatively accurate coverage in the classification examples and worse coverage in the regression examples, but do not explicitly report these results herein.

Scripts reproducing these experiments are available at \url{https://github.com/stephenbates19/nestedcv_experiments}.

\subsection{Classification}

\subsubsection{Low-dimensional logistic regression}
\label{subsec:lowd_log_exp}
We consider the logistic regression data generating model \ref{eq:log_regression_def} with $n=100$ observations and $p=20$ features, sampled as i.i.d. standard Gaussian variables. Due to the rotational symmetry of the features, the only parameter that affects behavior is the signal strength, and we explore models with Bayes error of either 33\% or 23\%. Here, we use (un-regularized) logistic regression as our fitting algorithm.
We report the results in Table~\ref{tab:lowd_d_logistic_summary}, finding that nested CV gives coverage much closer to the nominal level. Moreover, the point estimates have slightly less bias. We report the size of the NCV intervals relative to their CV counterparts per instance in Figure~\ref{fig:lowd_log_infl}.

\begin{table}[!h]
\small
\begin{center}
\begin{tabular}{| c | c || c | c || c | c | c | c || c : c | c : c | c : c |}
\hline
 
  \multicolumn{2}{|c||}{Setting} & \multicolumn{2}{c||}{Width} & \multicolumn{4}{c||}{Point estimates} &
 \multicolumn{6}{c|}{Miscoverage}\\
 
  \multicolumn{2}{|c||}{} & \multicolumn{2}{c||}{} & \multicolumn{4}{c||}{} &
 \multicolumn{2}{c}{CV} & \multicolumn{2}{c}{NCV}
 & \multicolumn{2}{c|}{DS} \\
 \multicolumn{1}{|c}{Bayes Error} & Target & 
 \multicolumn{1}{c}{NCV} & 
  \multicolumn{1}{c||}{DS} &
  \multicolumn{1}{c}{$\err$} & \multicolumn{1}{c}{CV} & \multicolumn{1}{c}{NCV} & \multicolumn{1}{c||}{DS} & 
 \multicolumn{1}{c}{Hi} & \multicolumn{1}{c}{Lo} &
 \multicolumn{1}{c}{Hi} & \multicolumn{1}{c}{Lo} & \multicolumn{1}{c}{Hi} & \multicolumn{1}{c|}{Lo} \\
 \hline

  33.2\% & $\err_{XY}$& 1.23 & 2.23 & 39.1\% & 39.6\% & 39.0\% & 40.1\% & 10\% & 8\% & 3\% & 5\% & 7\% & 6\%\\
 " & $\err$ & " & " & " & " & " & " & 9\% & 8\% & 3\% & 4\% & 6\% & 5\% \\
 \rule{0pt}{4ex} 
 
 22.5\% & $\err_{XY}$ & 1.47 & 2.25 & 28.7\% & 30.4\% & 28.1\% & 33.3\% & 11\% & 3\% & 4\% & 1\% & 16\% & 4\%\\
  " & $\err$ & " & " & " & " & " & " & 10\% & 2\% & 5\% & 0\% & 15\% & 3\%\\
 
\hline
\end{tabular}
\end{center}
\caption{Performance of cross-validation (CV), nested cross-validation (NCV), and data splitting with refitting (DS) in the low-dimensional logistic regression model from Section~\ref{subsec:lowd_log_exp}. Each row is a setting with a different signal strength, indexed by the Bayes error: the error of the true model. The nominal  total error rate is 10\%, i.e., 5\% above and below. A ``Hi" miscoverage is one where the confidence interval is too large and the point estimate falls below the interval; conversely for a ``Lo" miscoverage.
The standard error in each coverage estimate reported is about 0.5\%. The ``Target'' column indicates the target of coverage---the intervals are always generated identically, but we report the coverage of both $\err$ and $\err_{XY}$.}
\label{tab:lowd_d_logistic_summary}
\end{table}

Next, we return to the question of estimands, as in Section~\ref{sec:estimand_of_cv}. Since we do not have analytical results for the logistic regression model, we explore this in simulation. Here, we consider problems where the Bayes error rate is 22.5\%, and vary $n$ and $p$. We investigate two quantities. First, we investigate the correlation of $\errhat^\cv$ and $\err_{XY}$, and we find that it is small but larger than in the OLS case. See Figure~\ref{fig:lowd_log_target}.
Next, we check whether $\errhat$ has higher precision for $\err$ than for $\err_{XY}$. To this end, we compute the expected value of $|\errhat^\cv - \err|$ and of $|\errhat^\cv - \err_{XY}|$, and plot their relative difference in the right panel of Figure~\ref{fig:lowd_log_target}. We find that the CV point estimate is again slightly more precise as an estimate of $\err$ than of $\err_{XY}$ in this setting.

\begin{figure}
    \centering
    \includegraphics[height = 2in, trim = 0 0 60 0, clip]{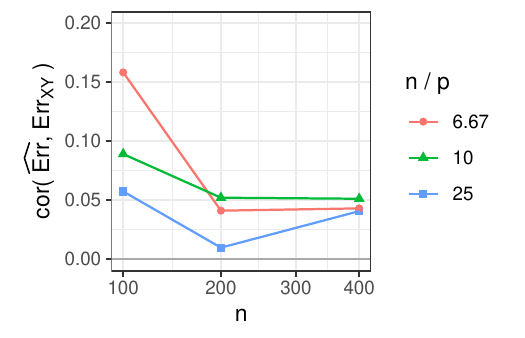}~
    \hspace{.2in}
    \includegraphics[height = 2in]{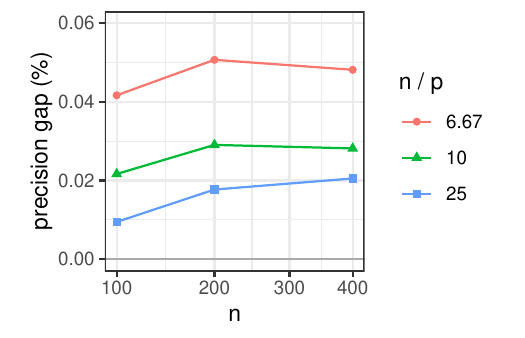}
    \caption{Behavior of cross-validation with a logistic regression model. Left: the correlation between the point estimate and instance-specific error. Right: fraction change in mean absolute deviation of the point estimate with respect to $\err_{XY}$ versus $\err$. }
    \label{fig:lowd_log_target}
\end{figure}

\subsubsection{High-dimensional sparse logistic regression}
\label{subsec:highd_logistic_exp}
We return to the high-dimensional logistic regression model introduced in Section~\ref{subsec:harrell_model}, generalizing slightly. We consider $n \in \{90, 200\}$ with $p = 1000$ features. The feature matrix has standard normal entries with an autoregressive covariance pattern such that adjacent columns have covariance $\rho$. In each case, we take $k=4$ nonzero entries of the covariance matrix and use sparse logistic regression. We report on the results in Table~\ref{tab:highd_logistic_big} and give the width\footnote{The width in Figure~\ref{fig:highd_logistic_infl} is reported relative to the version of cross-validation that holds out 2 folds at a time, since this is what is computed internally during NCV. In table Table~\ref{tab:highd_logistic_big} and elsewhere, we instead report widths relative to the usual $K$-fold CV.} in Figure~\ref{fig:highd_logistic_infl}. Again, NCV gives intervals with coverage much closer to the nominal level.

\begin{figure}[!h]
\captionsetup[subfigure]{justification=centering}
\centering
\begin{subfigure}[b]{.3\textwidth}
\centering
\includegraphics[height = 2in, trim = 30 0 0 0, clip]{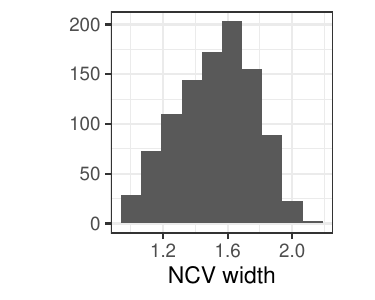}
\subcaption{$n=90, \rho = 0$}
\end{subfigure}
\hspace{.2cm}
\begin{subfigure}[b]{.3\textwidth}
\centering
\includegraphics[height = 2in, trim = 30 0 0 0, clip]{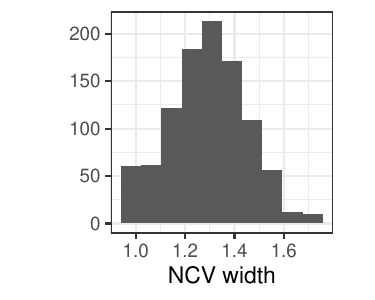}
\subcaption{$n=200, \rho = 0$}
\end{subfigure}
\hspace{.2cm}
\begin{subfigure}[b]{.3\textwidth}
\centering
\includegraphics[height = 2in, trim = 30 0 0 0, clip]{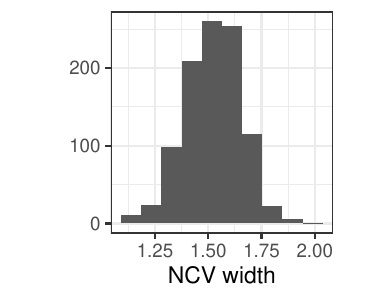}
\subcaption{$n=90, \rho = 0.5$}
\end{subfigure}
\hfill
\caption{Size of the nested CV intervals relative to the size of the na\"ive CV intervals
in the high-dimensional sparse logistic regression experiment from Section~\ref{subsec:highd_logistic_exp}.}
\label{fig:highd_logistic_infl}
\end{figure}

\begin{table}[!h]
\small
\begin{center}
\begin{tabular}{| c | c | c || c | c || c | c | c | c || c : c | c : c | c : c |}
\hline
 \multicolumn{3}{|c||}{Setting} & \multicolumn{2}{c||}{Width} & \multicolumn{4}{c||}{Point estimates} &
 \multicolumn{6}{c|}{Miscoverage}\\
 
  \multicolumn{3}{|c||}{} & \multicolumn{2}{c||}{} & \multicolumn{4}{c||}{} & 
 \multicolumn{2}{c}{CV} & \multicolumn{2}{c}{NCV} &
 \multicolumn{2}{c|}{DS} \\
 
 \multicolumn{1}{|c}{n} & \multicolumn{1}{c}{$\rho$} & Target & 
 \multicolumn{1}{c}{NCV} & \multicolumn{1}{c||}{DS} &
 \multicolumn{1}{c}{Bayes error} & \multicolumn{1}{c}{$\err$} & \multicolumn{1}{c}{CV} & \multicolumn{1}{c||}{NCV} & 
 \multicolumn{1}{c}{Hi} & \multicolumn{1}{c}{Lo} &
 \multicolumn{1}{c}{Hi} & \multicolumn{1}{c}{Lo} &
 \multicolumn{1}{c}{Hi} & \multicolumn{1}{c|}{Lo} \\
 \hline
  90 & 0 & $\err_{XY}$ & 1.53 & 2.24 & 22\% & 41.3\% & 41.8\% & 41.1\% & 16\% & 12\% & 6\% & 7\% & 9\% & 7\% \\
  " & " & $\err$ & " & " & " & " & " & " & 17\% & 13\% & 6\% & 8\% & 11\% & 9\%\\
   \rule{0pt}{4ex} 
   
  200 & 0 & $\err_{XY}$ & 1.66 & 2.26 & 22\% & 25.6\% & 26.7\% & 25.6\% & 14\% & 7\% & 3\% & 5\%& 9\% & 4\%\\
  " & " & $\err$ & " & " & " & " & " & " & 15\% & 7\% & 4\% & 6\% & 8\% & 5\%\\
  \rule{0pt}{4ex} 
  
  90 & 0.5 & $\err_{XY}$ & 1.80 & 2.25 & 13\% & 25.6\% & 27.5\% & 28.6\% & 20\% & 10\% & 5\% & 8\% & 15\% & 4\%\\
  " & " & $\err$ & " & " & " & " & " & " & 20\% & 11\% & 7\% & 9\% & 14\% & 3\% \\
\hline
\end{tabular}
\end{center}
\caption{Performance of cross-validation (CV), nested cross-validation (NCV), and data splitting (DS) in the high-d logistic regression model from Section~\ref{subsec:highd_logistic_exp}. The nominal (target) error rate is 10\%, i.e., 5\% above and below.  Other details as in Table~\ref{tab:lowd_d_logistic_summary}.}
\label{tab:highd_logistic_big}
\end{table}

\subsection{Regression}

\subsubsection{Low-dimensional linear model}
\label{subsec:lowd_lin_exp}

We next consider an OLS example. We take $X \in \mathbb{R}^{n\times p}$ with $p=20$ comprised of i.i.d. $\mathcal{N}(0,1)$. Further, we generate a response from the standard linear model:
\begin{equation*}
Y = X \theta + \epsilon
\end{equation*}
where $\epsilon$ is likewise i.i.d. $\mathcal{N}(0,1)$. We use OLS to estimate $\theta$. Note that by Lemma~\ref{lem:ols_equivariance}, the choice of $\theta$ does not affect the coverage rate of CV. The same argument shows that the choice of $\theta$ will not affect the coverage rate of nested CV, so we can take $\theta$ to be 0 without loss of generality. Similarly, both CV and NCV are unchanged when $X$ is transformed by an full-rank linear operator, so the results in this section would remain unchanged for Gaussian features with any full-rank correlation structure. 
We report the coverage of nested CV in Figure~\ref{fig:ols_covg_adj}. We find that this scheme works well and has good coverage for any $n$, overcovering somewhat for very small $n$. By contrast, na\"ive CV has poor coverage until $n$ is 400. In Figure~\ref{fig:ols_infl} we report on the width of the NCV intervals relative to their CV counterparts---the usual ratio is not that large for samples sizes of $n=100$ or greater.

\begin{figure}[!h]
\begin{center}
\includegraphics[height = 2.5in]{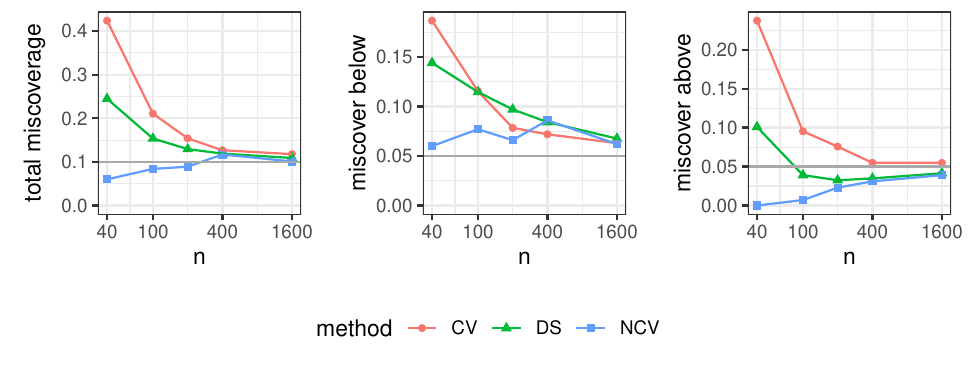}
\end{center}
\caption{Coverage of CV, data splitting, and nested CV in the OLS case.}
\label{fig:ols_covg_adj}
\end{figure}

\subsubsection{A high-dimensional sparse linear model}
\label{sec:subsec:highd_lin_exp}
We continue as in the previous experiment, but with $n \in \{50, 100\}$ and $p = 500$. We choose $\theta$ to have 4 nonzero entries of equal strength such that 
$$
\frac{\var(X \theta)}{\var(\epsilon)} = 4.
$$
Since $p > n$, we take the lasso estimator with a fixed penalty parameter. The parameter is chosen by minimizing the cross-validation estimate of prediction error on a single independent run, and then this value is fixed for the experimental replicates. We report on the results in Table~\ref{tab:highd_linear}, again finding the NCV has better coverage than CV, although both struggle when $n=50$.
The ratio the nested CV interval width to the na\"ive CV interval width is relatively stable across observations, see Figure~\ref{fig:highd_lasso_infl}.

\begin{table}[!h]
\small
\begin{center}
\begin{tabular}{| c | c || c | c || c | c | c | c | c || c : c | c : c | c : c |}
\hline
 
  \multicolumn{2}{|c||}{Setting} & \multicolumn{2}{c||}{Width} & \multicolumn{5}{c||}{Point estimates} &
 \multicolumn{6}{c|}{Miscoverage}\\
 
  \multicolumn{2}{|c||}{} & \multicolumn{2}{c||}{} & \multicolumn{5}{c||}{} &
 \multicolumn{2}{c}{CV} & \multicolumn{2}{c}{NCV}
 & \multicolumn{2}{c|}{DS} \\
 \multicolumn{1}{|c}{n} & Target & 
 \multicolumn{1}{c}{NCV} & 
  \multicolumn{1}{c||}{DS} &
 \multicolumn{1}{c}{Bayes Error} & \multicolumn{1}{c}{$\err$} & \multicolumn{1}{c}{CV} & \multicolumn{1}{c}{NCV} & \multicolumn{1}{c||}{DS} & 
 \multicolumn{1}{c}{Hi} & \multicolumn{1}{c}{Lo} &
 \multicolumn{1}{c}{Hi} & \multicolumn{1}{c}{Lo} & \multicolumn{1}{c}{Hi} & \multicolumn{1}{c|}{Lo} \\
 \hline


  50  & $\err_{XY}$ & 1.27 & 2.33 & 1 & 2.37 & 2.61 & 2.37 & 3.08 & 14\% & 8\% & 9\% & 3\% & 8\% & 14\%\\
  "  & $\err$ & " & " & " & " & " & " & " & 9\% & 7\% & 5\% & 1\% & 6\% & 16\% \\
   \rule{0pt}{4ex} 
   
  100  & $\err_{XY}$ & 1.89 & 2.20 & 1 & 1.54 & 1.61 & 1.55 & 1.63 & 13\% & 14\% & 3\% & 2\% & 4\% & 11\% \\
  "  & $\err$ & " & " & " & " & " & " & " & 14\% & 16\% & 4\% & 4\% & 5\% & 12\% \\
\hline
\end{tabular}
\end{center}
\caption{Performance of cross-validation (CV),nested cross-validation (NCV), and data splitting (DS) in the high-d linear regression model. The nominal (target) error rate is 10\%, i.e., 5\% above and below.  Other details as in Table~\ref{tab:lowd_d_logistic_summary}.}
\label{tab:highd_linear}
\end{table}

\section{Real data examples}
\label{sec:real_data}
Lastly, we evaluate the nested CV procedure on real data sets from the UCI repository \citep{Dua:2019}. In each case, we repeatedly subsample a small number of observations, perform nested CV on the subsample, and then use the many remaining observations to determine the accuracy of the fitted model. We consider the following data sets:
\begin{itemize}
    \item {\bf Communities and crimes (CC)}. This data set is comprised of measurements of 1994  communities in the US. We predict the crime rate of each community, a real number normalized to be between 0 and 1, based on 99 demographic features of the community.
    \item {\bf Crop mapping (crp)}. This data set is comprised of optical radar measurements of cropland in Manitoba in 2012. We filter the data set to contain two classes, corn and oats, and then do binary classification based on 174 features. Here, we add a small amount of label noise so that the best possible classifier has a misclassification rate of about 5\%.
\end{itemize}
We again use sparse linear or logistic regression as our fitting algorithm. The results are reported in Table~\ref{tab:real_data_summary}. We find that nested CV generally has coverage that is much closer to the nominal rate than na\"ive CV. Data splitting has poor coverage in this case due to the small sample size, but is significantly better with $ = 100$ samples than with $n = 50$ samples.

\begin{table}[!h]
\small
\begin{center}
\begin{tabular}{| c | c | c || c | c || c | c | c | c || c : c | c : c | c : c |}
\hline
 \multicolumn{3}{|c||}{Setting} & \multicolumn{2}{c||}{Width} & \multicolumn{4}{c||}{Point estimates} &
 \multicolumn{6}{c|}{Miscoverage}\\
 
  \multicolumn{3}{|c||}{} & \multicolumn{2}{c||}{} & \multicolumn{4}{c||}{} & 
 \multicolumn{2}{c}{CV} & \multicolumn{2}{c}{NCV} &
 \multicolumn{2}{c|}{DS} \\
 
 \multicolumn{1}{|c}{data} & \multicolumn{1}{c}{$n$} & Target & 
 \multicolumn{1}{c}{NCV} & \multicolumn{1}{c||}{DS} &
 \multicolumn{1}{c}{$\err$} & \multicolumn{1}{c}{CV} & \multicolumn{1}{c}{NCV} &  \multicolumn{1}{c||}{DS} &
 \multicolumn{1}{c}{Hi} & \multicolumn{1}{c}{Lo} &
 \multicolumn{1}{c}{Hi} & \multicolumn{1}{c}{Lo} &
 \multicolumn{1}{c}{Hi} & \multicolumn{1}{c|}{Lo} \\
 \hline

 
 CC & 50 & $\err_{XY}$ & 2.82 & 1.77 & 0.029 & 0.031 & 0.029 & .034 & 4\% & 20\% & 1\% & 13\% & 1\% & 33\%\\
 " & " & $\err$ & " & " & " & " & " & " & 2\% & 22\% & 0\% & 12\% & 1\% & 37\% \\
 \rule{0pt}{4ex} 
 
 CC & 100 & $\err_{XY}$ & 1.46 & 2.05 & 0.023 & 0.024 & 0.023 & .025 & 4\% & 13\% & 2\% & 7\% & 1\% & 24\%\\
 " & " & $\err$ & " & " & " & " & " & " & 2\% & 12\% & 1\% & 4\% & 1\% & 24\%\\
 \rule{0pt}{4ex} 
 
 crp & 50 & $\err_{XY}$ & 1.21 & 1.89 & 10.6\% & 10.7\% & 10.6\% & 10.8\% & 6\% & 8\% & 2\% & 6\% & 4\% & 31\% \\
 " &  " & $\err$ & " & " & " & " & " & " & 7\% & 12\% & 3\% & 8\% & 2\% & 31\%\\
 \rule{0pt}{4ex} 
 
 crp &  100 & $\err_{XY}$ & 1.52 & 2.00 & 9.5\% & 9.7\% & 9.5\% & 9.4\% & 6\% & 6\% & 4\% & 5\% & 4\% & 15\%\\
 " &  " & $\err$ & " & " & " & " & " & " & 8\% & 9\%  & 5\%  & 7\% & 4\% & 15\% \\
\hline
\end{tabular}
\end{center}
\caption{Performance of cross-validation (CV),nested cross-validation (NCV), and data splitting (DS) with the real data sets.
The nominal (target) error rate is 10\%, i.e., 5\% above and below. Other details as in Table~\ref{tab:lowd_d_logistic_summary}.}
\label{tab:real_data_summary}
\end{table}

\section{Discussion}
Our investigation had two main main components. First, we discussed point estimates of prediction error via subsampling techniques. Our primary result is that common estimates of prediction error---cross-validation, bootstrap, data splitting, and covariance penalties---should be viewed as estimates of the \emph{average} prediction error, averaged across other hypothetical data sets from the same distribution. 
The formal results here were all for the special case of the linear model using unregularized OLS for model-fitting, although we also saw similar behavior in simulation for logistic regression; see Figure~\ref{fig:lowd_log_target}. A further important question is how regularization affects this behavior. In an additional experiment, we find that $\errhat$ does track $\err_{XY}$, albeit weakly, when there is regularization; see Appendix~\ref{app:regularization_cor}. We look forward to future work explaining the behavior of cross-validation and other estimates of prediction error in these settings.
 
Second, we discussed inference for cross-validation, deriving an estimator for the MSE of the CV point estimate, nested CV. The nested CV scheme has consistently superior coverage compared to na\"ive cross-validation confidence intervals, which makes it an appealing choice for providing confidence intervals for prediction error. 
Nonetheless, we wish to be clear that nested CV is more computationally intensive than standard CV---we use about $1000$ times more model fits per example because of the repeated splitting. For example, in the logistic regression example from Section~\ref{subsec:harrell_model}, nested CV takes about $10$ seconds on a personal computer. 

A fundamental open question is to understand under what conditions the standard CV intervals will be badly behaved, making the nested CV computations necessary. Roughly speaking, we expect the standard CV intervals to  perform better when $n / p$ is larger and when more regularization is used. In our experiments, we saw that even in the mundane linear model with $n / p = 10$, the miscoverage rate of standard CV was about 50\% larger than the nominal rate. As $n$ increases, however, the violation decreases. Moreover, the asymptotic results in \citet{Austern} and \citet{bayle2020} show that the coverage is correct in the $p$ fixed, $n \to \infty$ limiting regime. The stability conditions therein may also be able to shed light on settings with small samples or high-dimensions. We look forward to future work in this direction.
 
To conclude, we point out several additional future directions. First, one could adapt nested CV to cases with dependent data, as is done for standard CV \citep{rabinowicz2020cross}. Second, we note that the general leave-out style strategy of cross-validation can also be used to ``fill-in'' data for downstream use. Examples include {\em pre-validation} \citep{TE2002, hofling2008} and {\em cross-fitting}  \cite[e.g,][]{newey2018crossfitting}). Further, confidence intervals for prediction accuracy are used to evaluate variable importance \citep{williamson2021nonparametric, zhang2020floodgate}. We suspect that our nested cross-validation proposal could be adapted to improve the accuracy of these and related approaches.
Lastly, cross-validation is often used to compare regression procedures (such as when selecting the value of a tuning parameter); see Section~\ref{subsec:related_work}. We anticipate that nested CV can be extended to give valid confidence intervals for the difference in prediction error between two models.


\medskip
\medskip

An \texttt{R} package implementing nested CV is available at \url{https://github.com/stephenbates19/nestedcv}.

\section*{Acknowledgements}
The authors would like to acknowledge Frank Harrell for a seminar  and personal correspondence,  alerting them to the miscoverage of cross-validation in the high-dimensional logistic regression model. We would like to thank  Michael Celentano, Ryan Tibshirani, Larry Wasserman, Lester Mackey, Adam Smoulder, Alexandre Bayle and three anonymous reviewers/editors for helpful comments on earlier versions of this manuscript.  We would especially like to thank Bradley Efron for his comments on this work and for his fundamental contributions to this general area over the last 40 years.  S.~B.~was partially supported by a Ric Weiland Graduate Fellowship. 
T.H. was partially supported by grants DMS-2013736 and IIS
1837931 from the National Science Foundation, and grant 5R01 EB
001988-21 from the National Institutes of Health.
R.T. was supported by the National Institutes of Health (5R01
EB001988-16) and the National Science Foundation (19 DMS1208164).

\newpage
\bibliographystyle{apalike}
\bibliography{cv_cis}

\appendix
\renewcommand\thefigure{\thesection.\arabic{figure}} \renewcommand\thetable{\thesection.\arabic{table}}

\section{Results for bootstrap estimates of prediction error}
\label{app:bootstrap_theorems}

In this section, we give results parallel to those in Section~\ref{sec:estimand_of_cv} for bootstrap estimates of prediction error. We consider two bootstrap estimators: the .632 estimator and the out-of-bag (OOB) error estimator \citep{efron1983estimating, Efron1997, Breiman1996OOB}. In the interest of brevity, we reference the definitions of these two estimators from \citet{Efron1997}: we will use $\errhat^\oob$ to denote the out-of-bag estimator defined in equation (17) from \citet{Efron1997} and $\errhat^\boot$ to denote the .632 estimator defined in equation (24) from \citet{Efron1997}. In both cases, we will assume that the underlying model fitting routine is OLS.

\begin{lemma}
\label{lem:ols_boot_equivariance}
When using OLS as the fitting algorithm, $\errhat^\boot$ and $\errhat^\oob$ are linearly invariant.
\end{lemma}

Thus, by the results in Section~\ref{sec:estimand_of_cv}, we conclude that bootstrap should also be viewed as an estimator of $\err$ or $\err_X$, rather than of $\err_{XY}$.

\begin{remark}
We expect these properties also hold for most other resampling-based estimators of prediction error, in addition to those explicitly considered herein. For example, nested CV is linearly invariant.
\end{remark}

\section{Data splitting}
\label{app:data_splitting}

\subsection{Data splitting without refitting}
We first consider \emph{data splitting without refitting}. We partition the data into disjoint sets $\Ical^\train \cup \Ical^\test = \{1,\dots,n\}$, and fit the model on the training set $\hat{\theta}^\train = \A((X_i, Y_i)_{i \in \Ical^\train})$. As in cross-validation, we then let $e_i = \ell(\hat{f}(x_i, \hat{\theta}^\train), y_i)$ for $i \in \Ical^\test$. We can then estimate the prediction error as
\begin{equation}
\errhat^\splitt := \frac{1}{|\Ical^\test|}\sum_{i \in \Ical^\test} e_i
\label{eq:data_split_def}
\end{equation}
and give a valid estimate of its standard error as
\begin{equation*}
\sehat^\splitt := \sqrt{\frac{1}{|\Ical^\test| - 1}\sum_{i \in \Ical^\test} (e_i - \errhat^\splitt)^2}.
\end{equation*}
We define $\widetilde{X} = (X_i)_{i\in\Ical^\train}$ and $\widetilde{Y} = (Y_i)_{i\in\Ical^\train}$. Then $\errhat^\splitt$ is an unbiased estimate for 
\begin{equation}
    \err_{\widetilde{X}\widetilde{Y}} := \E \left[\ell\left(\hat{f}(X_{n+1}, \hat{\theta}^\train), Y_{n+1}\right)\mid (\widetilde{X}, \widetilde{Y})\right],
\label{eq:err_xytilde_def}
\end{equation}
where the expectation is only over a fresh test point $(X_{n+1}, Y_{n+1})$.
The advantage of this approach is that one can accurately estimate the prediction error and provide valid inference. A first disadvantage is the prediction error and standard error estimates are valid for the model trained only on the subset $\Ical^\train$. Secondly, the estimates of prediction error have reduced precision because they rely only on the subset $\Ical^\test$. Thus, confidence intervals for prediction error may be much wider than those from CV. We confirm this in experiments in Section~\ref{sec:experiments}.

\subsection{Data splitting with refitting}
A slightly different form of data splitting is also commonly used, which we will refer to as \emph{data splitting with refitting}. In this approach, one follows the same steps as above to obtain the estimator $\errhat^\splitt$ from \eqref{eq:data_split_def}, but then conducts a final refitting step to obtain $\hat{\theta} = \A(X,Y)$, the model fit on the full data. Here, one deploys the model on the full data, $\hat{f}(\cdot, \hat{\theta})$---the idea is that the model on the full data is superior to the model fit using only the training subset. 
Data-splitting with refitting is similar to cross-validation, and analogs of our results from Section~\ref{sec:estimand_of_cv} carry over to this setting. 
\begin{lemma}
The estimator $\errhat^\splitt$ is linearly invariant.
\end{lemma}
Thus, the conclusions of Theorem~\ref{thm:ols_cv_independent_err}, Corollary~\ref{cor:errx_lower_mse}, Corollary~\ref{cor:cor_errx_errxy_highd}, and Corollary~\ref{cor:err_errxy_gap_highd} hold for $\errhat^\splitt$. In particular, $\errhat^\splitt$ has lower error for estimating $\err$ than for estimating $\err_{XY}$, and $\errhat^\splitt$ is asymptotically uncorrelated with $\err_{XY}$. 
Moreover, the standard error estimate $\sehat^\splitt$ becomes invalid when refitting, and asymptotically it is too small, as stated next.
\begin{proposition}
In the setting of Theorem~\ref{thm:ols_cv_independent_err},
\begin{equation}
    \E\left[\left(\sehat^\splitt \right)^2\right] = \E\left[\left(\errhat^\splitt - \err \right)^2\right] - \underbrace{\var(\err_{\widetilde{X}\widetilde{Y}})}_{\ge 0} - \underbrace{\left(\err - \E[\err_{\widetilde{X}\widetilde{Y}}]\right)^2}_{\ge 0}.
\label{eq:data_split_se_expansion}
\end{equation}
\label{prop:data_split_se_expansion}
\end{proposition}
This is a non-asymptotic result. In the proportional asymptotic limit~\eqref{eq:highd_limit_def}, for both sides of~\eqref{eq:data_split_se_expansion} above all terms are of order $1 / n$, except for the final term on the right-hand side, which is of constant order.\footnote{It is of constant order when the test set $I^\test$ is a constant fraction of $n$. This can be made smaller by choosing the ratio of $I^\test$ to $n$ converge to 0 at some rate, but then data splitting will not even achieve a $1/\sqrt{n}$ rate of precision.} 
This means that the standard error estimate derived from data splitting is too small, and confidence intervals based on this number will have coverage that is too small. Importantly, this behavior is \emph{not} only due to the difference in sample size used to fit $\hat{\theta}^\splitt$ and $\hat{\theta}$. The final term in~\eqref{eq:data_split_se_expansion} is the result of the difference in sample size, but even without this term, the estimate is too small asymptotically (due to the middle term on the right-hand side of~\eqref{eq:data_split_se_expansion}). The standard error estimate $\sehat^\splitt$ is similarly too small if one wishes to estimate $(\errhat^\splitt - \err_{XY})^2$; see Proposition~\ref{prop:data_split_se_expansion2} in Appendix~\ref{app:additional_results}. In Figure~\ref{fig:ds_covg_ols}, we verify with an experiment that the data splitting intervals do not have coverage approaching the nominal level, even as $n$ and $p$ grow. (Na\"ive cross-validation has a similar miscoverage problem---see Figure~\ref{fig:cv_ols_covg1} and Figure~\ref{fig:cv_ols_covg2}.) We emphasize, however, that the data splitting without replacement intervals have coverage much closer to the nominal level than cross-validation because they do not suffer from the data re-use (i.e., correlation across folds) problem; see again Figure~\ref{fig:ds_covg_ols} and Figure~\ref{fig:cv_ols_covg1}. The fact that data splitting with refitting produces intervals that are slightly too small is an interesting theoretical occurrence, but, based on our simulations, the magnitude of the bias is small enough that it is not a major worry in practice. In summary, data splitting with refitting has partially similar behavior to cross-validation; the point estimate has higher precision as an estimate of estimate of average prediction error and the standard error estimate is slightly too small in the proportional asymptotic limit. 

\section{Bias estimation}
\label{app:bias_estimation}
The NCV estimate of prediction error is unbiased for $\err$ for the procedure with a reduced sample size of $n(K-2)/K$, but this will be typically be slightly biased upwards for $\err$ with the full sample size. This discrepancy can be estimated by running both the usual $K$-fold CV and nested CV. An unbiased estimated for the difference in $\err$ at a sample of size  $n (K-2) / K$ (the sample size used in nested CV) to a sample of size $n (K-1) / K$ (the sample size used in standard CV) is
\begin{equation*}
    \errhat^\ncv - \errhat^\cv.
\end{equation*}

Now, since we expect that prediction error scales as $a + b / n$ in $n$ for some unknown constants $a$ and $b$ (the parametric rate), an estimator for the difference in in $\err$ at a sample of size  $n$  to $\err$ at a sample of size $n (K-2) / K$ (the sample size of each model used in nested CV) is then:
\begin{equation}
\label{eq:bias_estimator}
    \widehat{\bias} := \left(1 + \left(\frac{K-2}{K}\right)\right) \left(\errhat^\ncv - \errhat^\cv\right).
\end{equation}
The left term in the sum, ``$1$'', accounts for the bias when going from size $n (K-2) / K$ to size $n (K-1) /K$ and the right term in the sum, accounts for the bias going from a sample size from $n (K-1) / K$ to size $n$; this scaling due to the form of $a + b / n$. Combining this with the estimate of MSE in the previous section leads to the confidence intervals in~\eqref{eq:finalConfInt}.

\section{Proofs}
\label{app:proofs}

\begin{lemma}
\label{lem:invariant_ols_residuals}
A linearly invariant estimator $\errhat$ of prediction error is function of the residuals after running OLS fitting. That is, with $r = Y - X \hat{\theta} \in \mathbb{R}^n$ where $\theta$ is the OLS estimate, we have that 
\begin{equation*}
    \errhat((x_1, y_1), \dots, (x_n, y_n)) = g(r)
\end{equation*}
for some function $g : \mathbb{R}^n \to \mathbb{R}$.
\end{lemma}

\begin{proof}[Proof of Lemma~\ref{lem:invariant_ols_residuals}]
Let $\hat{\theta}$ be the OLS estimate. By the invariance property, we see that 
\begin{equation*}
    \errhat((x_1, y_1), \dots, (x_n, y_n)) = \errhat((x_1, y_1 - \hat{\theta}^\top x_1), \dots, (x_n, y_n - \hat{\theta}^\top x_n)),
\end{equation*}
and the right side of the above is a function only of $r$.
\end{proof}

\begin{proof}[Proof of Lemma~\ref{lem:ols_equivariance}]
Without loss of generality, consider point 1. We will show that the residual from CV on point 1 is the same when using either the data $(x_1, y_1), \dots, (x_n, y_n)$ or the data $(x_1, y_1'), \dots, (x_n, y_n')$, where $y_i' = y_i + x_i^\top \kappa$.

Let $\cI_1 \subset \{1,\dots,n\}$ be the indices in the same fold as observation 1. Let $\hat{\theta}_{(-1)}$ be the OLS estimate of $\theta$ based on only the points $(x_i, y_i)_{i \notin \cI_1}$---all the points not in the same fold as point 1---and let $\hat{\theta}'_{(-1)}$ be the OLS estimate of $\theta$ based on only the points $(x_i, y'_i)_{i \notin \cI_1}$. Then 
\begin{align*}
    \hat{\theta}_{(-1)} &= \argmin_{\theta} \sum_{i \notin \cI_1}(y_i - x_i^\top \theta)^2 \\
    &= \argmin_{\theta} \sum_{i \notin \cI_1}(y_i + \kappa^\top x_i - x_i^\top (\theta + \kappa))^2 \\
    &= -\kappa + \argmin_{\theta} \sum_{i \notin \cI_1}(y_i' - x_i^\top \theta)^2 \\
    &= \hat{\theta}'_{(-1)} - \kappa.
\end{align*}
Where the third equality is a result of changing variables from $\theta$ to $\theta + \kappa$. As a result, 
\begin{equation*}
y_1 - x_1^\top \hat{\theta}_{(-1)} = y_1' - x_1^\top \hat{\theta}'_{(-1)}.
\end{equation*}
The CV estimate of prediction error is the mean of the squared residuals, so since by the preceding display the residuals are the same for either the original or shifted data, the CV estimate of prediction error is the same in each case.
\end{proof}

\begin{proof}[Proof of Corollary~\ref{cor:errx_lower_mse}]
\begin{align*}
    \E[(\errhat^\cv - \err_{XY})^2] &= \E[(\errhat^\cv - \err_{X} + \err_{X} - \err_{XY})^2] \\
        &= \E[(\errhat^\cv - \err_{X})^2] + \E[(\err_{X} - \err_{XY})^2] + \E[(\errhat^\cv - \err_{X})(\err_{X} - \err_{XY})] \\
        &= \E[(\errhat^\cv - \err_{X})^2] + \E[(\err_{X} - \err_{XY})^2] \\ 
        &\qquad \qquad +\E\left[\E[(\errhat^\cv - \err_{X})(\err_{X} - \err_{XY})] \mid X\right] \\
        &= \E[(\errhat^\cv - \err_{X})^2] + \E[(\err_{X} - \err_{XY})^2]
\end{align*}
Where the last equality follows from Theorem~\ref{thm:ols_cv_independent_err}. The result follows by noting that the second term in the final line is $\E\left[\var(\err_{XY}\mid X)\right]$.
\end{proof}

\begin{proof}[Proof of Theorem~\ref{thm:err_and_err_x_higd}]

Note that the residuals from OLS remain unchanged when we transform the features by a full-rank linear transformation. Thus, we can transform the features by $\Sigma^{-1/2}$ without changing the behavior of invariant estimators, and so without loss of generality we can take $\Sigma$ to be the identity matrix. Nonetheless, we leave $\Sigma$ generic throughout the proofs.

We first give an expression for $\err_{XY}$.
\begin{align*}
\err_{XY} &= \E_{X_{n+1}, Y_{n+1}} \left(X_{n+1}^\top \hat{\beta} - Y_{n+1}\right)^2 \\
    &= \sigma^2 +  \E_{X_{n+1}}\left(X_{n+1}^\top (\hat{\beta} - \beta)\right)^2 \\
    &= \sigma^2 + \E_{X_{n+1}} \left[ (\hat{\beta} - \beta)^\top X_{n+1}X_{n+1}^\top (\hat{\beta} - \beta)\right] \\
    &= \sigma^2 + \norm{\Sigma^{1/2}(\hat{\beta} - \beta)}^2 \\
    &= \sigma^2 + \frac{\sigma^2}{n}\norm{\Sigma^{1/2}(X^T X / n )^{-1/2}Z}^2 \\
    &= \sigma^2 + \frac{\sigma^2}{n}\norm{\Sigma^{1/2}\hat{\Sigma}^{-1/2} Z}^2,
\end{align*}
where $Z \sim \mathcal{N}(0,I)$ is a function only of the noise $\epsilon$. Next, we decompose the variance with the conditional variance formula
\begin{align*}
    \var(\err_{XY}) &= \E\left[\var(\err_{XY} \mid X)\right] + \var\left( \E\left[\err_{XY} \mid X\right] \right)\\
    &= \underbrace{\E\left[\var(\err_{XY} \mid X)\right]}_{\text{var due to $Y \mid X$}} + \underbrace{\var(\err_X)}_{\text{var due to $X$}},
\end{align*}
which was previously stated in \eqref{eq:err_xy_var_decomp} of the main text. We next derive asymptotic rates for the two components of this sum. 

\paragraph{Proof of first claim.}
We begin with the first term, which corresponds to the variance cause by the randomness in $Y \mid X$.
\begin{align*}
    \E\left[\var(\err_{XY} \mid X)\right]
    &= \E \ \var\left( \frac{\sigma^2}{n}\norm{\Sigma^{1/2}\hat{\Sigma}^{-1/2} Z}^2 \mid X\right) \\
    &= \frac{\sigma^4}{n^2} \E \ \var(\lambda_1^2 Z_1^2+ \dots + \lambda_p^2 Z_p^2 \mid X) \\
    &= \frac{2\sigma^4}{n^2} \E \left[\lambda_1^4 + \dots + \lambda_p^4 \right]
\end{align*}
where $\lambda_1^2,\dots,\lambda_p^2$ are the eigenvalues of $\Sigma \hat{\Sigma}^{-1}$, a function of $X$. Thus, the proof is complete once we show that the expectation term in the final line is $\Theta(n)$, which we turn to next. 

Notice that
\begin{align*}
    \E[\lambda_1^4+\dots+\lambda_p^4] &= \E\left[\tr\left((\Sigma \hat{\Sigma}^{-1})(\Sigma \hat{\Sigma}^{-1})\right)\right] \\
    &= (c_1 + c_2)n^2 p + c_2 n^2 p^2 \qquad \qquad \text{(by Corollary 3.1 of \citet{vanRosen1988moments})},
\end{align*}
where
\begin{equation*}
    c_1 := \frac{(n-p-2)}{(n-p)(n-p-1)(n-p-3)} \qquad c_2 := \frac{1}{(n-p)(n-p-1)(n-p-3)}
\end{equation*}
We conclude 
\begin{equation*}
\E \left[\lambda_1^4 + \dots + \lambda_p^4 \right] = \Theta(n).
\end{equation*}
This completes the proof of the first claim.

\paragraph{Proof of second claim.}
Now, we turn our attention to the second term, which corresponds to the variance caused by the randomness in $X$.
\begin{align*}
\var(\err_X) &= \var\left(\frac{\sigma^2}{n} \tr (\Sigma \hat{\Sigma}^{-1})\right) \\
    &= \frac{\sigma^4}{n^2} \var\left(\tr (\Sigma \hat{\Sigma}^{-1})\right) \\
    &= \frac{\sigma^4}{n^2} \cdot \Theta(1) \qquad \text{ as } n,p \to \infty
\end{align*}
where the final equality comes from the second moments of the inverse Wishart distribution \citep{vanRosen1988moments}. To elaborate on the last equality,
\begin{align*}
    \var\left(\tr (\Sigma \hat{\Sigma}^{-1})\right) &= p \cdot \var((\Sigma \Sigma^{-1})_{11}) + p(p-1)\cov((\Sigma \Sigma^{-1})_{11}, (\Sigma \Sigma^{-1})_{22}) \\
    &= p \cdot \frac{2n^2}{(n-p-1)^2(n-p-3)} + p(p-1) \cdot \frac{2n^2}{(n - p)(n - p - 1)^2(n - p - 3)} \\
    &= \Theta(1).
\end{align*}

\end{proof}

\begin{proof}[Proof of Corollary~\ref{cor:cor_errx_errxy_highd}]
For the first claim,
\begin{align*}
    \cor(\err_X, \err_{XY}) 
    &= \frac{\cov(\err_X, \err_{XY})}{\sqrt{\var(\err_X)\var(\err_{XY}})} \\
    &= \frac{\cov\left(\err_X, \E[\err_{XY} \mid X]\right)}{\sqrt{\var(\err_X)\var(\err_{XY}})} \\
    &= \frac{\cov\left(\err_X, \err_X\right)}{\sqrt{\var(\err_X)\var(\err_{XY}})} \\
    &= \sqrt{\frac{\var(\err_X)}{\var(\err_{XY})}} \to 0 \quad \text{as } n \to \infty.
\end{align*}
The second equality above comes from the conditional covariance formula, conditioning on $X$.

For the second claim, note that $\errhat = g(\err_X, U)$ for an independent $U \sim \text{unif}[0,1]$ for some function $g$ (by Theorem~\ref{thm:ols_cv_independent_err}). That is, $\errhat$ is a random function of $\err_X$. As a result, 
\begin{align*}
    \cor(\err_{XY}, \errhat) 
    &= \cor(\err_{XY}, g(\err_X, U)) \\
    &= \frac{\cov(\err_{XY}, g(\err_X, U))}{\sqrt{\var(\err_{XY}) \var(g(\err_X, U))}} \\
    &= \frac{\cov(\err_{X}, \E[g(\err_X, U) \mid X])}{\sqrt{\var(\err_{XY}) \var(g(\err_X, U))}} \\ 
    &= \sqrt{\frac{\var(\err_X)}{\var(\err_{XY})}} \cdot
    \frac{\cov(\err_{X}, \E[g(\err_X, U) \mid X])}{\sqrt{\var(\err_{X}) \var(g(\err_X, U))}} \\ 
    &\le \sqrt{\frac{\var(\err_X)}{\var(\err_{XY})}} \cdot
    \cor(\err_X, \E[g(\err_X, U)\mid X]) \\ 
    &\le \sqrt{\frac{\var(\err_X)}{\var(\err_{XY})}}  \to 0  \quad \text{as } n \to \infty.
\end{align*}
The third equality above comes from the conditional covariance formula, conditioning on $X$.
\end{proof}

\begin{proof}[Proof of Corollary~\ref{cor:err_errxy_gap_highd}]
For the first claim,
\begin{align*}
    \E\left[\left(\errhat - \err_{XY}\right)^2\right] 
    &= \E\left[\left(\errhat - \err_X + \err_X - \err_{XY}\right)^2\right] \\
    &= \E\left[\left(\errhat - \err_X\right)^2 +\left(\err_X - \err_{XY}\right)^2 - 2\cdot\left(\errhat - \err_X\right)\cdot\left(\err_X - \err_{XY}\right)\right] \\
    &= \E\left[\left(\errhat - \err_X\right)^2 +\left(\err_X - \err_{XY}\right)^2\right] \qquad \qquad  ( \text{Theorem~\ref{thm:ols_cv_independent_err}}) \\
    &= \E\left[\left(\errhat - \err_X\right)^2\right] + \var(\err_X) +\var(\err_{XY}) - 2\cdot \cov(\err_X,\err_{XY}) \\
    &\ge \E\left[\left(\errhat - \err_X\right)^2\right] +\var(\err_{XY}) - 2\sqrt{\var(\err_{XY})
    \cdot {\var(\err_X)}} \\
    &= \E\left[\left(\errhat - \err_X\right)^2\right] + \frac{1}{n}\cdot \Omega(1).
\end{align*}

For the third claim,
\begin{align*}
    \E\left[\left(\errhat - \err\right)^2\right] - \E\left[\left(\errhat - \err_X\right)^2\right]
    &= \E\left[\left(\errhat - \err_X + \err_X - \err\right)^2\right] - \E\left[\left(\errhat - \err_X\right)^2\right] \\
    &= \var(\err_X) + 2\cdot\cov(\errhat - \err_X, \err_X) \\
    &= \frac{1}{n^2} \cdot O(1) + 2\cdot\cov(\errhat, \err_X)
\end{align*}
so
\begin{align*}
    \left\lvert \E\left[\left(\errhat - \err\right)^2\right] - \E\left[\left(\errhat - \err_X\right)^2\right] \right\rvert
    &\le \frac{1}{n^2} \cdot O(1) + \frac{2}{n}\cdot\sqrt{\var(\errhat)} \cdot O(1).
\end{align*}

The second claim follows by combining these two results.
\end{proof}

\begin{proof}[Proof of Proposition~\ref{prop:data_split_se_expansion}]
\begin{align*}
    \E\left[\left(\sehat^\splitt \right)^2\right] &=
        \E\left[\var(\errhat^\splitt | \widetilde{X}, \widetilde{Y})\right] & \text{(def. of $\errhat^\splitt$)} \\
        &= \var(\errhat^\splitt) - \var(\E[\errhat^\splitt \mid \widetilde{X}, \widetilde{Y}]) \\
        &= \E\left[\left(\errhat^\splitt - \E[\err_{\widetilde{X}, \widetilde{Y}}]\right)^2\right] - \var(\err_{ \widetilde{X}\widetilde{Y}}) \\
        &= \E\left[\left(\errhat^\splitt - \err\right)^2\right] -  \var(\err_{ \widetilde{X}\widetilde{Y}}) - \left(\err - \E[\err_{\widetilde{X}\widetilde{Y}}]\right)^2. \\
\end{align*}
\end{proof}

\begin{proof}[Proof of Proposition~\ref{prop:errx_cov_penalty}]
Let $j$ be drawn uniformly on $\{1,\dots,n\}$. For $i=1,\dots,n$, let $Y_i'$ be an independent draw from the distribution of $Y_i \mid X_i$. Then
\begin{align*}
    \err_X
    &= \err_\inn(X) + \E\left[[(Y_{n+1} - \hat{f}(X_{n+1}, \hat{\theta})^2  - (Y_{j}' - \hat{f}(X_{j}, \hat{\theta})^2 \right] \\ 
    &= \err_\inn(X) + \E\left[(\hat{\theta} - \theta)^\top \Sigma (\hat{\theta} - \theta)\right] -
        \E\left[(\hat{\theta} - \theta)^\top \hat{\Sigma} (\hat{\theta} - \theta)\right] \\
    &= \err_\inn(X) + \E\left[\E\left[(\hat{\theta} - \theta)^\top \Sigma (\hat{\theta} - \theta) \mid X \right]\right] - \frac{p\sigma^2}{n} \\
    &= \err_\inn(X) + \frac{\sigma^2}{n} \tr(\hat{\Sigma}^{-1}\Sigma) - \frac{p\sigma^2}{n}
\end{align*}
\end{proof}

\begin{proof}[Proof of Corollary~\ref{cor:randomx_cov_penalty}]
The first claim follows from a Jensen-type inequality for matrices in \citet{groves1969note}. The second claim is a result of the mean of the inverse-Wishart distribution.
\end{proof}

\begin{proof}[Proof of Theorem~\ref{thm:ncv_unbiased}]
For the first part of the theorem, consider without loss of generality the first entry of the vector $\texttt{a\_list}$ in Algorithm~\ref{alg:nested_cv}. This is the term (a) in \eqref{eq:holdout_se_identity}. Similarly, the first entry entry of \texttt{b\_list} is unbiased for (b) in \eqref{eq:holdout_se_identity}. The result follows.
\end{proof}

\begin{proof}[Proof of Lemma~\ref{lem:ols_boot_equivariance}]
We consider $\errhat^\oob$ and $\errhat^\boot$. Consider a single bootstrap sample $\Ical \subset \{1,\dots,n\}$, a multiset of cardinality $n$. We will show that both the residuals and out-of-bag error are identical when using either the data $(x_1, y_1), \dots, (x_n, y_n)$ or the data $(x_1, y_1'), \dots, (x_n, y_n')$, where $y_i' = y_i + x_i^\top \kappa$.

Let $\hat{\theta}$ be the OLS estimate for $\theta$ based on the bootstrap sample $\Ical$ with points $(x_1, y_1), \dots, (x_n, y_n)$, and let $\hat{\theta}'$ be the corresponding estimate for the data $(x_1, y_1'), \dots, (x_n, y_n')$. Then
\begin{align*}
    \hat{\theta} &= \argmin_{\theta} \sum_{i \in \Ical}(y_i - x_i^\top \theta)^2 \\
    &= \argmin_{\theta} \sum_{i \in \Ical}(y_i + \kappa^\top x_i - x_i^\top (\theta + \kappa))^2 \\
    &= -\kappa + \argmin_{\theta} \sum_{i \in \Ical}(y_i' - x_i^\top \theta)^2 \\
    &= \hat{\theta}' - \kappa.
\end{align*}
As a result 
\begin{equation*}
y_i - x_i^\top \hat{\theta} = y_i' - x_i^\top \hat{\theta}'
\end{equation*}
for $i=1,\dots,n$. Both the out-of-bag and .632 bootstrap estimators are functions only of the quantities $y_i - x_i^\top \hat{\theta}$, across many different bootstrap samples, so the proof is complete.
\end{proof}

\section{Additional technical results}
\label{app:additional_results}

\begin{lemma}
In the setting of Proposition~\ref{prop:data_split_se_expansion}, with notation as in Section~\ref{subsec:data_splitting},
\begin{equation*}
\left(\err - \E[\err_{\widetilde{X}\widetilde{Y}}]\right)^2
 = \Theta(1).
\end{equation*}
\label{lem:sample_size_gap}
\end{lemma}

\begin{proof}[Proof of Lemma~\ref{lem:sample_size_gap}]
Using the notation from the proof of Theorem~\ref{thm:err_and_err_x_higd},
\begin{align*}
\err &= \sigma^2 + \E\left[\frac{\sigma^2}{n}\norm{\Sigma^{1/2}\hat{\Sigma}^{-1/2} Z}^2\right] \\
&= \sigma^2 + \frac{\sigma^2}{n}\E\left[\tr(\Sigma \hat{\Sigma}^{-1})\right] \\
&= \sigma^2 + \frac{\sigma^2 p}{n - p -1}.
\end{align*}
The result follows.
\end{proof}

\begin{proposition}
In the setting of Theorem~\ref{thm:err_and_err_x_higd}, with notation as in Section~\ref{subsec:data_splitting}, suppose also that $\Ical^\train > p$ and $|\Ical^\test| / n \to c \in (0,1)$.
\begin{equation}
    \E\left[\left(\sehat^\splitt \right)^2\right] = \E\left[\left(\errhat^\splitt - \err_{XY} \right)^2\right] - \var(\err_{\widetilde{X}\widetilde{Y}})
    - \var(\err_{XY})
    - \left(\err - \E[\err_{\widetilde{X}\widetilde{Y}}]\right)^2 + O(1/n^2).
\label{eq:data_split_se_expansion2}
\end{equation}
\label{prop:data_split_se_expansion2}
\end{proposition}
In the proportional asymptotic limit~\eqref{eq:highd_limit_def}, for both sides of~\eqref{eq:data_split_se_expansion2} above all terms are of order $1 / n$, except the $O(1/n^2)$ term on the right-hand side.

\begin{proof}[Proof of 
Proposition~\ref{prop:data_split_se_expansion2}]
From Proposition~\ref{prop:data_split_se_expansion}, we have
\begin{equation*}
    \E\left[\left(\sehat^\splitt \right)^2\right] = \E\left[\left(\errhat^\splitt - \err \right)^2\right] - \var(\err_{\widetilde{X}\widetilde{Y}}) - \left(\err - \E[\err_{\widetilde{X}\widetilde{Y}}]\right)^2.
\end{equation*}
We now expand the first term:
\begin{align*}
    \E\left[\left(\errhat^\splitt - \err \right)^2\right]
    &= \E\left[\left(\errhat^\splitt - \err_{XY} + \err_{XY} - \err\right)^2\right] \\
    &= \E\left[\left(\errhat^\splitt - \err_{XY}\right)^2\right] + \var(\err_{XY}) + \\
    &\qquad 2\E\left[\left(\errhat^\splitt - \err_{XY}\right)\cdot(\err_{XY} - \err)\right] \\
    &= \E\left[\left(\errhat^\splitt - \err_{XY}\right)^2\right] + \var(\err_{XY}) + \\
    &\qquad 2 \cdot \cov\left(\errhat^\splitt - \err_{XY}, \err_{XY} - \err\right) \\
    &= \E\left[\left(\errhat^\splitt - \err_{XY}\right)^2\right] - \var(\err_{XY}) + \\
    &\qquad 2 \cdot \cov\left(\errhat^\splitt, \err_{XY}\right) \\
\end{align*}
Thus, it remains to show only that  $\cov\left(\errhat^\splitt, \err_{XY}\right) = o(1/n)$. To this end,
\begin{align*}
    \cov\left(\errhat^\splitt, \err_{XY}\right) = \cov\left(\err_{\widetilde{X}}, \err_X\right)
\end{align*}
by the conditional covariance decomposition (conditioning on $X$) and applying Theorem~\ref{thm:ols_cv_independent_err}. Applying Cauchy-Schwarz,
\begin{equation*}
    \cov\left(\err_{\widetilde{X}}, \err_X\right) \le \sqrt{\var(\err_X)\cdot\var(\err_{\widetilde{X}})} = O(n^{-2}),
\end{equation*}
where in the last equality we applied Theorem~\ref{thm:err_and_err_x_higd}.
\end{proof}

\section{Further simulation results}
\setcounter{figure}{0}  

\subsection{The bias of CV in the proportional asymptotic regime.}
In the setting of Figure~\ref{fig:highd_rates}, we report on the variance and bias squared of cross-validation in Figure~\ref{fig:ols_rates_bv}. We find that for the values of $n$ we consider, the variance is much larger than the bias. We see, however, that for larger $n$, the constant bias would overtake the variance as the leading issue. (One way to address this issue is to have the number of folds grow appropriately.)

\begin{figure}
    \centering
    \includegraphics[height=2.75in]{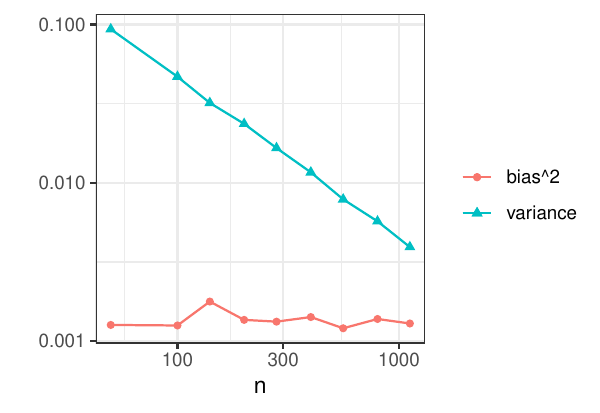}
    \caption{The bias and variance of the CV estimate of $\err$. The simulation uses the same setup as Figure~\ref{fig:highd_rates}; we use OLS as the fitting algorithm, the data is generated from a linear model with Gaussian errors, and $n/p = 5$. We us $K=10$ folds.}
    \label{fig:ols_rates_bv}
\end{figure}

\subsection{Compute times}
\label{app:compute_times}
In Table~\ref{tab:compute_times}, we report on the runtime of CV and NCV for our experiments. 

\begin{table}[!ht]

\begin{center}
\begin{tabular}{| c | c | c | c |}
\hline
Section & Experiment & CV time (s) & NCV time (s) \\
\hline
\ref{subsec:lowd_log_exp} & Bayes Error 33.2\% & 0.02  & 12.2 \\
" & Bayes Error 22.5\% & 0.02  & 13.5 \\
\ref{subsec:highd_logistic_exp} & $n = 90, \rho = 0$ & 0.16 & 6.3 \\
" & $n = 200, \rho = 0$ & 0.2 & 17 \\
" & $n = 90, \rho = 0.5$ & 0.2 & 11 \\
\ref{subsec:lowd_lin_exp} & $n=40$ & 0.02 & 6.6 \\
" & $n=100$ & 0.01 & 7.1 \\
"  & $n=200$ & 0.01 & 7.8 \\
" & $n=400$ & 0.02 & 10 \\
" & $n=1600$ & 0.02 & 22 \\
\ref{sec:subsec:highd_lin_exp} & $n=50$ & 0.1 &  44 \\
" & $n=100$ & 0.1 &  46 \\
\ref{sec:real_data} & CC, $n=50$ & 0.02 & 15 \\
 " & CC, $n=100$ & 0.02 & 15 \\
 " & crp, $n=50$ & 0.2 & 22 \\
 " & crp, $n=100$ & 0.2 & 24 \\
\hline
\end{tabular}
\end{center}
\caption{Approximate computation times for one run of CV and NCV for each of the experimental settings.}
\label{tab:compute_times}
\end{table}

\subsection{Additional details on experiments from Section~\ref{sec:experiments}}
This section reports additional results from the set of experiments in Section~\ref{sec:experiments}.

\begin{figure}[!h]
\captionsetup[subfigure]{justification=centering}
\centering
\begin{subfigure}[b]{.45\textwidth}
\centering
\includegraphics[height = 2in, trim = 30 0 0 0, clip]{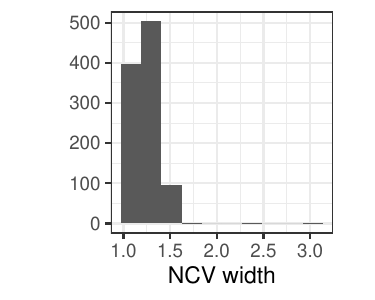}
\subcaption{Bayes error 33.2\%}
\end{subfigure}
\hspace{.2cm}
\begin{subfigure}[b]{.45\textwidth}
\centering
\includegraphics[height = 2in, trim = 30 0 0 0, clip]{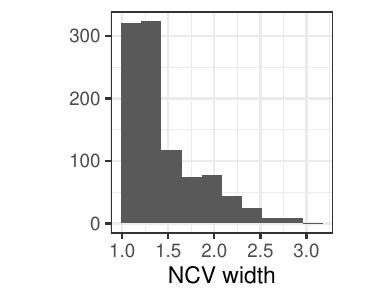}
\subcaption{Bayes error 22.5\%}
\end{subfigure}
\hfill
\caption{Size of the nested CV intervals relative to the size of the na\"ive CV intervals in the low-dimensional logistic regression example from Section~\ref{subsec:lowd_log_exp}.}
\label{fig:lowd_log_infl}
\end{figure} 

\begin{figure}[!h]
\begin{center}
\includegraphics[height = 2in]{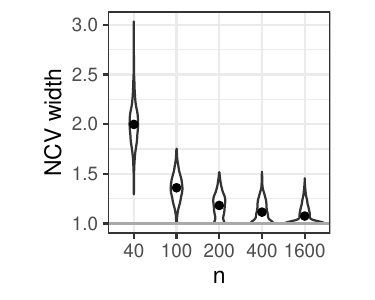}
\end{center}
\caption{Width of nested CV intervals relative to the width of the na\"ive CV intervals in the experiment from Section~\ref{subsec:lowd_lin_exp}}
\label{fig:ols_infl}
\end{figure}

\begin{figure}[!h]
\captionsetup[subfigure]{justification=centering}
\centering
\begin{subfigure}[b]{.45\textwidth}
\centering
\includegraphics[height = 2in, trim = 30 0 0 0, clip]{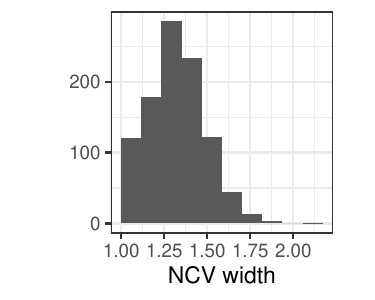}
\subcaption{$n=50$}
\end{subfigure}
\hspace{.2cm}
\begin{subfigure}[b]{.45\textwidth}
\centering
\includegraphics[height = 2in, trim = 30 0 0 0, clip]{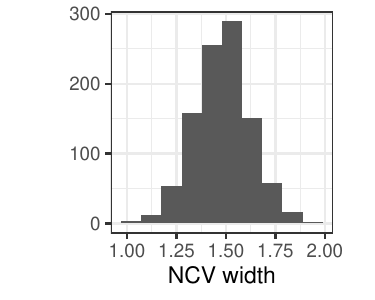}
\subcaption{$n=100$}
\end{subfigure}
\hfill
\caption{Size of the nested CV intervals relative to the size of the na\"ive CV intervals in the high-dimensional sparse regression example from Section~\ref{sec:subsec:highd_lin_exp}.}
\label{fig:highd_lasso_infl}
\end{figure} 

\subsection{Number of folds}
\label{subapp:num_folds_exp}
We next investigate how the number of folds affects the \emph{CV inflation}: the ratio of the true standard error of the point estimate compared to the CV estimate of standard error. We consider a linear model with $p = 20$ features that are sampled as i.i.d. standard Gaussians. The number of observations ranges from $50$ to $400$. We use OLS as the fitting algorithm, so as a result of Lemma~\ref{lem:ols_equivariance}, the results do not depend on the true coefficients $\theta$. We report on the results in Figure~\ref{fig:ols_choice_k}, where we find that the number of folds has minimal impact on the inflation, although more folds gives moderately better coverage for small $n$. We also find that even when $n / p$ is as large as $20$, there is appreciable CV inflation, and na\"ive cross-validation leads to intervals with poor coverage.

\begin{figure}[!h]
\begin{center}
\includegraphics[height = 2.5in]{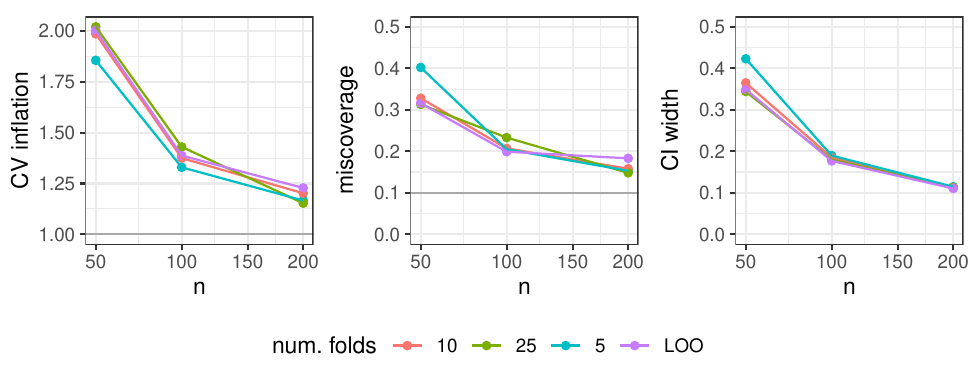}
\end{center}
\caption{The CV inflation and coverage of the na\"ive interval in the low-dimensional linear model with 20 features. The dark grey horizontal line in the middle panel gives the nominal miscoverage level. The right panel gives the average width of the confidence interval.}
\label{fig:ols_choice_k}
\end{figure}

\subsection{CV in the proportional region}
We next show some experimental results further exploring the regime from Section~\ref{subsec:err_errx_results}. In Figure~\ref{fig:highd_rates2}, we compare the accuracy of CV when covering $\err$ compared to its accuracy when covering $\err_{XY}$, and we find that it has higher accuracy for the former, by a constant fraction as $n,p\to \infty$. We also see that $\err_{XY}$ and $\errhat^\cv$ are essentially uncorrelated.

Next, we plot the coverage of CV in Figure~\ref{fig:cv_ols_covg1}, and see that it is far from the nominal level, even as $n$ and $p$ grow and the intervals have oracle debiasing so that they are centered around the correct value. In Figure~\ref{fig:cv_ols_covg2}, we show that these intervals have higher than the nominal miscoverage rate in both tails, confirming that the miscoverage rate is not due to bias---it is due to the intervals being too narrow.

\begin{figure}[!t]
    \centering
    \begin{subfigure}[b]{3.5in}
    \centering
    \includegraphics[height = 2.5in]{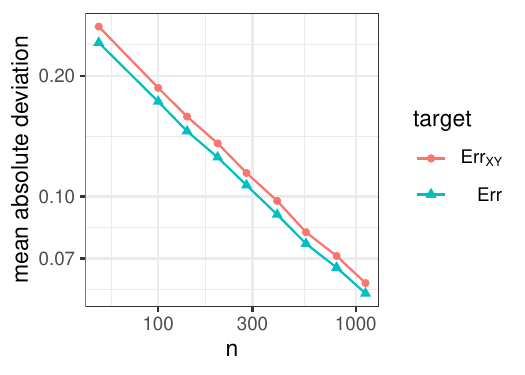}
    \end{subfigure}
    \begin{subfigure}[b]{2.75in}
    \centering
    \includegraphics[height = 2.5in]{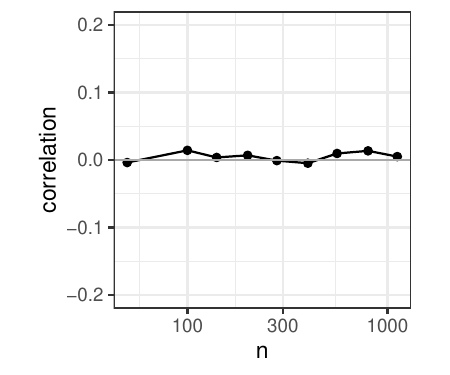}
    \end{subfigure}
    \caption{Simulation results comparing the error of CV when estimating $\err$ to its error when estimating $\err_{XY}$. Left: the mean absolute deviation between $\errhat^\cv$ and $\err$ or $\err_{XY}$. Right: $\cor(\errhat^\cv, \err_{XY})$}
        \label{fig:highd_rates2}
\end{figure}

\begin{figure}
    \centering
    \includegraphics{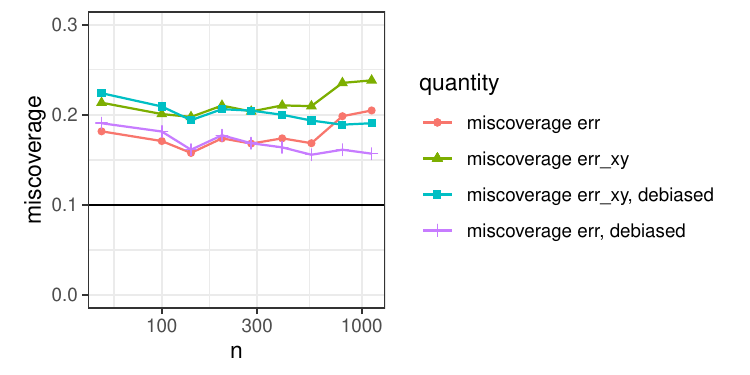}
    \caption{Coverage of CV intervals with OLS in the proportional regime. Nominal coverage rate is 10\%.}
    \label{fig:cv_ols_covg1}
\end{figure}

\begin{figure}
    \centering
    \includegraphics{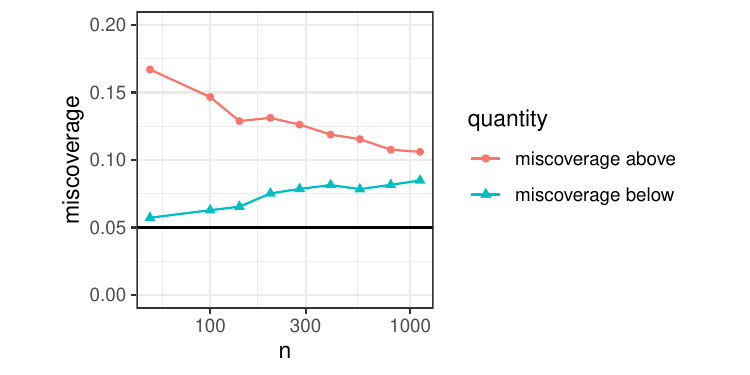}
    \caption{Coverage of $\err_{XY}$ of CV intervals with OLS and oracle debiasing in the proportional regime. Nominal coverage rate is 5\% above and below.}
    \label{fig:cv_ols_covg2}
\end{figure}

\subsection{Coverage of data splitting in proportional regime}
Here, we record further results about data splitting in the setting from Section~\ref{subsec:data_splitting}. In Figure~\ref{fig:ds_covg_ols}, we show that data splitting does not approach the nominal coverage rate, even when the intervals are debiased by an oracle to be centered at the correct value.

\begin{figure}[ht]
    \centering
    \includegraphics[height = 2.5in]{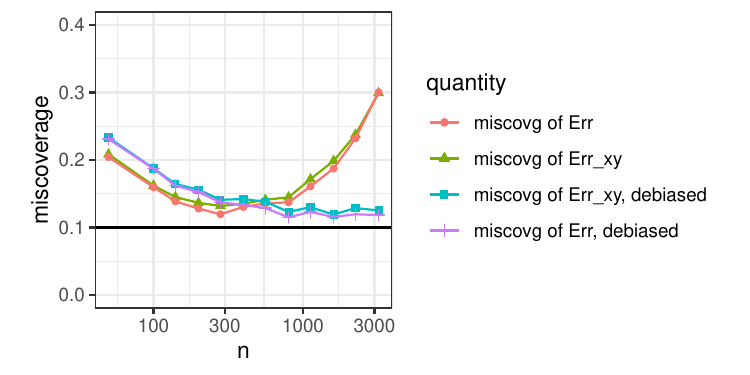}
    \caption{Coverage rate of data splitting with OLS. Nominal coverage rate is 10\%.}
    \label{fig:ds_covg_ols}
\end{figure}

\subsection{.632 bootstrap empirical influence function CIs}
We next discuss standard error estimates with the bootstrap point estimates introduced previously in Section~\ref{app:bootstrap_theorems}.
The OOB estimate of prediction error has has an associated estimate of standard error, based on estimates of the empirical influence functions \citep{efron1983estimating, Efron1997}. In particular, from \citet{Efron1997} we use the point estimate from (24) and the SE estimate from (36). To get an estimate of the standard error for the .632 estimator, we re-scale the estimated OOB standard error by the ratio of (24) to (17) of \citet{Efron1997}, as suggested therein.
We investigate the coverage of these intervals on our real-data examples and report on the results in Table~\ref{tab:real_data_summary_expand}. We find that the .632 confidence intervals are generally acceptable, with reasonable coverage. The intervals are typically, but not always, wider than the NCV intervals. The bootstrap point estimates are typically more biased that the NCV point estimates.

\newgeometry{margin=1in} 
\begin{landscape}

\begin{table}[!ht]
\footnotesize
\begin{center}
\begin{tabular}{| c | c | c || c | c | c | c | c | c | c | c || c : c | c : c | c : c | c : c |}

 \hline
 \multicolumn{3}{|c||}{Setting} & \multicolumn{8}{c||}{Width and Point Estimates} &
 \multicolumn{8}{c|}{Miscoverage}
 \\ \hline
  \multicolumn{3}{|c||}{} & \multicolumn{3}{c|}{Width} &
   \multicolumn{1}{c|}{} & \multicolumn{4}{c||}{Estimate} &
 \multicolumn{2}{c|}{CV} & \multicolumn{2}{c|}{NCV} & \multicolumn{2}{c|}{.632} & \multicolumn{2}{c|}{OOB} \\
 
 \multicolumn{1}{|c}{Data} & \multicolumn{1}{c}{$n$} & Target & 
 \multicolumn{1}{c}{NCV} & \multicolumn{1}{c}{.632} & 
 \multicolumn{1}{c|}{OOB} &
 \multicolumn{1}{c|}{$\err$} & \multicolumn{1}{c}{CV} & \multicolumn{1}{c}{NCV} & \multicolumn{1}{c}{.632} &
 \multicolumn{1}{c||}{OOB} &
 \multicolumn{1}{c}{Hi} & \multicolumn{1}{c|}{Lo} & \multicolumn{1}{c}{Hi} & \multicolumn{1}{c|}{Lo} &
 \multicolumn{1}{c}{Hi} & \multicolumn{1}{c|}{Lo} &\multicolumn{1}{c}{Hi} & \multicolumn{1}{c|}{Lo}\\
 \hline
 
 CC &  50 & $\err_{XY}$ & 2.82 & 3.06 & 4.0 & .029 & .031 & .029 & .047 & .067 & 4\% & 20\% & 1\% & 13\% & 0\% & 9\% & 2\% & 5\%\\
 " &  " & $\err$ & " & " & " & " & " & " & " & " & 2\% & 22\% & 0\% & 12\% & 0\% & 9\% & 1\% & 3\% \\
 \rule{0pt}{4ex} 
 
 CC &  100 & $\err_{XY}$ & 1.46 & 1.82 & 2.22 & .023 & .024 & .023 & .025 & .03 & 4\% & 13\% & 2\% & 7\% & 1\% & 8\% & 3\% & 4\% \\
 " &  " & $\err$ & " & " & " & " &  " & " & " & " & 2\% & 12\% & 1\% & 4\% & 0\% & 6\% & 2\% & 2\% \\
 \rule{0pt}{4ex} 
 
 crp &  50 & $\err_{XY}$ & 1.21 & 1.55 & 1.87 &10.5\% & 10.7\% & 10.5\% & 12.0\% & 14.4\% & 7\% & 8\% & 2\% & 5\% & 1\% & 6\% & 1\% & 4\% \\
" &  " & $\err$ & " & " & " & " & " & " & " & " & 6\% & 8\% & 3\% & 8\% & 1\% & 7\% & 2\% & 4\% \\
 \rule{0pt}{4ex}

 crp &  100 & $\err_{XY}$ & 1.10 & 1.20 & 1.32 & 9.5\% & 9.7\% & 9.6\% & 10.1\% & 11.1\% & 5\% & 7\% & 3\% & 6\% & 1\% & 8\% & 2\% & 5\% \\
 " &  " & $\err$ & " & " & " & " & " & " & " & " & 7\% & 7\%  & 4\%  & 6\% & 1\% & 8\% & 2\% & 5\% \\

\rule{0pt}{4ex} 
 S Lgstc &  90 & $\err_{XY}$ & 1.54 & 1.27 & 1.44 & 41.3\% & 41.8\% & 40.9\% & 40.9\% & 46.5\% & 17\% & 13\% & 7\% & 7\% & 4\% & 3\% & 19\% & 1\% \\
 " &  " & $\err$ & " & " & " & " & " & " & " & " & 17\% & 14\%  & 7\%  & 7\% & 2\% & 0\% & 11\% & 0 \%\\

\hline
\end{tabular}
\end{center}
\caption{Performance of various methods with the real data sets. ``S. Logistic'' is the sparse logistic model from Section~\ref{subsec:harrell_model}. Other details as in Table~\ref{tab:lowd_d_logistic_summary}.}
\label{tab:real_data_summary_expand}
\end{table}

\end{landscape}
\restoregeometry

\subsection{Number of Repeated Splits}
In our experiments, we found that a large number (e.g., 200) of random splits of nested CV were needed to obtain stable estimates of the standard error. In Figure~\ref{fig:harrell_num_reps}, we show how the estimate of standard error relative to the na\"ive CV estimate (\emph{inflation}) converges as the number of repetitions increases for one example of the logistic regression model in Section~\ref{subsec:harrell_model}.


\begin{figure}[!h]
\begin{center}
\includegraphics[height = 2.5in]{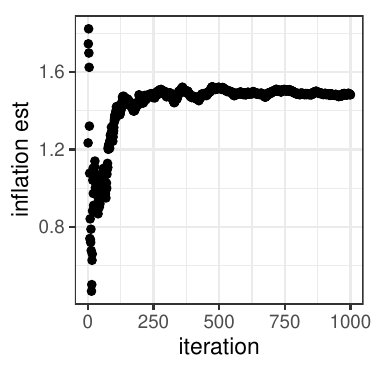}
\end{center}
\caption{The nested CV inflation estimate after a large number of repeated splits.}
\label{fig:harrell_num_reps}
\end{figure}

\subsection{The Austern and Zhou estimate of Standard Error}
\label{subapp:az_se_estimator}
Next, we investigate the standard error estimator of~\citet{Austern}, which we will call the \emph{AZ estimator}. Under certain stability conditions, that work proves a triangular array central limit theorem for the CV point estimate around $\err$, and the AZ estimator is shown to be asymptotically correct. We investigate this estimator in the OLS setup from Section~\ref{subsec:lowd_lin_exp}. In particular, we compare the average size of the confidence intervals based on the AZ estimate of standard error to those from na\"ive CV, reporting the results in Figure~\ref{fig:ols_infl_az}. Comparing the $p = 20$ curve (teal) to that from nested CV in Figure~\ref{fig:ols_infl}, we see that the AZ estimator produces significantly larger intervals than nested CV---they are a factor of two larger for $n=200$. (Nested CV has correct coverage in this experiment, see Figure~\ref{fig:ols_covg_adj}.) Still, the AZ estimator is not too conservative once $n = 800$ for $p = 20$ or $n = 200$ when $p = 5$, in line with the asymptotic results proved in~\citet{Austern}.

\begin{figure}
    \centering
    \includegraphics[height = 2.75in]{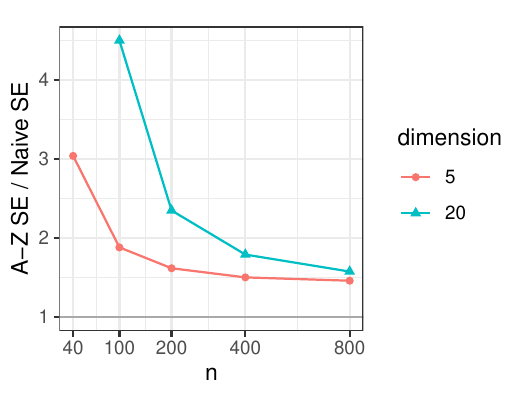}
    \caption{The average size of the confidence intervals of the AZ estimator in the linear model.}
    \label{fig:ols_infl_az}
\end{figure}

\subsection{The effect of regularization}
\label{app:regularization_cor}
In this simulation, we check the results of Section~\ref{sec:estimand_of_cv} in the presence of regularization. In particular, we return to the linear model simulation setup of Section~\ref{subsec:lowd_lin_exp} and check the correlation between $\errhat$ and $\err_{XY}$ with the ridge regression estimator using varying regularization strength. Section~\ref{sec:estimand_of_cv} shows that this correlation is nearly zero with no regularization, so we are interested in how this changes for ridge regression. We report on the results in Figure~\ref{fig:highd_linear_cor}. We find that for low levels of regularization, the correlation remains small. For larger values of regularization, the correlation can be as large as 15\% in this example. Our intuition for this behavior is that with regularization, there is less overfitting, so that the CV estimator has fresher data, so to speak, so it is able to partially track $\err_{XY}$. We look forward to future work understanding this behavior.

\begin{figure}
    \centering
    \includegraphics[height = 2.75in]{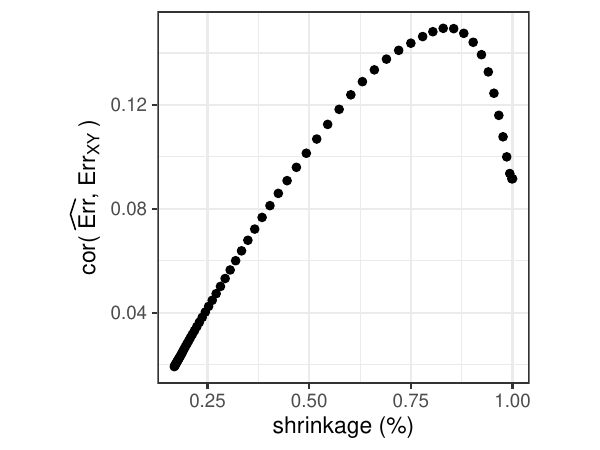}
    \caption{The correlation of $\errhat$ and $\err_{XY}$ with ridge regression. The regularization strength varies from small (left) to large (right). To report on the level of regularization in an interpretable way, the x-axis shows the fraction of reduction in norm of the fitted coefficient vector of the ridge regression fit compared to the unregularized fit.}
    \label{fig:highd_linear_cor}
\end{figure}

\section{Variance stabilizing transformation for 0-1 loss}
\label{app:variance_stabilizing_transform}
For classification settings with 0-1 loss, we can get better confidence intervals using a variance stabilizing transformation \cite[e.g.,][]{YU2009vst}. If we observe an empirical miscoverage level of $\bar{e} \in [0,1]$ from $n$ samples, we form a confidence interval for $\text{sin}^{-1}(\sqrt{\err})$ as
\begin{equation*}
    \left(\text{sin}^{-1}(\sqrt{\bar{e}}) - z_{1-\alpha/2} \cdot  \sqrt{\frac{1}{4n}}, \quad \text{sin}^{-1}(\sqrt{\bar{e}}) + z_{1-\alpha/2} \cdot \sqrt{\frac{1}{4n}} \right),
\end{equation*}
and invert the transformation to get a confidence interval for $\err$ (and analogously for $\err_{XY}$ or other targets). With nested CV, we inflate the intervals by looking at the ration between na\"ive CV and nested CV:
\begin{equation*}
    \left(\text{sin}^{-1}\left(\sqrt{\errhat^\ncv}\right) - z_{1-\alpha/2} \cdot \frac{\sqrt{\msehat}}{\sehat} \cdot \sqrt{\frac{1}{4n}}, \quad \text{sin}^{-1}\left(\sqrt{\errhat^\ncv}\right) + z_{1-\alpha/2} \cdot \frac{
    \sqrt{\msehat}}{\sehat} \cdot \sqrt{\frac{1}{4n}} \right).
\end{equation*}
(I.e., the transformed intervals are inflated by the amount that was estimated on the original scale.) Recall that $\sqrt{\msehat}$ is the nested CV estimated of the width of the confidence interval, whereas $\sehat$ is the na\"ive CV estimate of the width of the confidence interval.

\section{A low-dimensional asymptotic analysis}
\label{app:lowd_asymptotics}
In this section, we present a complementary asymptotic analysis to that of Section~\ref{sec:estimand_of_cv} in a more traditional asymptotic regime where $p$ is fixed and $n \to \infty$. Here, we observed similar a similar behavior; $\err_X$ is close toe $\err$, so we expect CV to have better accuracy for $\err$ than for $\err_{XY}$. In this case, however, the behavior is all of order higher than $1 / \sqrt{n}$. This means that for sufficiently large $n$, the difference between the various estimands is negligible compared to the variance of the CV estimator. Moreover, here $\err$ and $\err_{XY}$ approach the Bayes error at rate $1/n$, so asymptotically one has equal accuracy for estimating $\sigma^2, \err$ or $\err_{XY}$; see \citet{Wager2019} for a related discussion. Thus,  in this section, the phenomena is only observable in higher-order asymptotics terms. See Figure~\ref{fig:lowd_asymptotics_viz} for a visualization of the various rates.

\begin{figure}[!t]
    \centering
    \includegraphics[page = 2, width = 4in, trim = 0 32cm 38cm 0, clip]{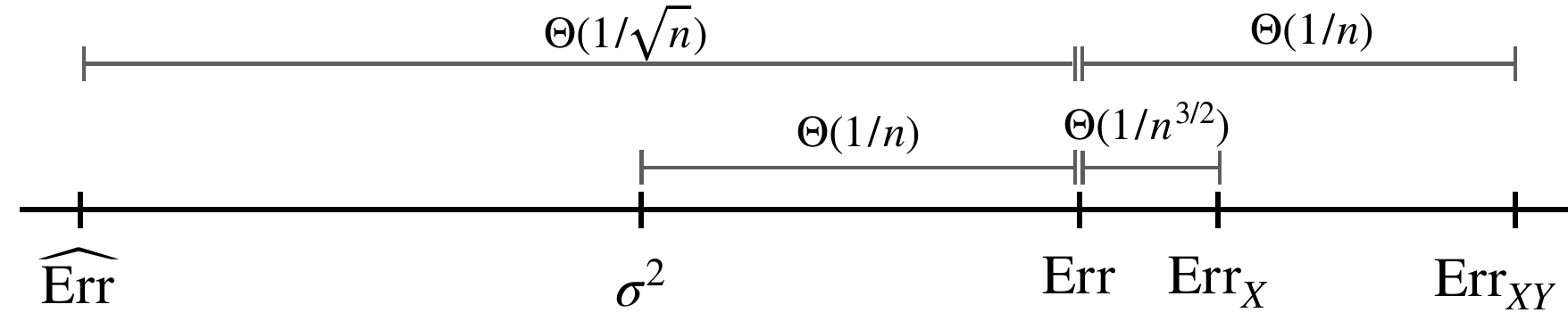}
    \caption{The relationship among various notions of prediction error in fixed $p$, $n \to \infty$ asymptotic limit. Recall that $\sigma^2$ is the Bayes error: the error rate of the best possible model}
    \label{fig:lowd_asymptotics_viz}
\end{figure}

\begin{theorem}
Suppose the homoskedastic Gaussian linear model in \eqref{eq:linear_model_assumption} holds and that we use squared-error loss. In addition, assume that feature vectors $X_i \sim \mathcal{N}(0, \Sigma_p)$ for any full-rank $\Sigma_p$. Then, as $n \to \infty$ with $p$ fixed, we have
\begin{equation*}
    \E_{X}\left[\var(\err_{XY} \mid X)\right] = \Theta(1 / n^2)
\end{equation*}
and
\begin{equation*}
    \var(\err_X) = \E(\err_X - \err)^2 = \Theta(1 / n^3).
\end{equation*}
\label{thm:err_and_err_x_lowd}
\end{theorem}
The proof is included at the end of this section. This result can readily be extended beyond the case of Gaussian features, but we do not pursue this at this time. 
From this result, we can extract the following
\begin{corollary}
In the setting of Theorem~\ref{thm:err_and_err_x_lowd},
\begin{equation*}
    \cor(\err_{XY}, \err_X) \to 0 \quad \text{as } n \to \infty.
\end{equation*}
Moreover, for any linearly invariant estimator $\errhat$ (such as $\errhat^\cv$ using OLS as the fitting algorithm),
\begin{equation*}
    \cor(\err_{XY}, \errhat) \to 0 \quad \text{as } n \to \infty.
\end{equation*}
\label{cor:cor_errx_errxy_lowd}
\end{corollary}
Which means that CV not tracking $\err_{XY}$.
The proof of this is as in the proof of Corollary~\ref{cor:cor_errx_errxy_highd}, applying Theorem~\ref{thm:err_and_err_x_lowd} in place of Theorem~\ref{thm:err_and_err_x_higd}.

\begin{proof}[Proof of Theorem~\ref{thm:err_and_err_x_lowd}]
The proof follows as in the proof of Theorem~\ref{thm:err_and_err_x_higd}, noting that in this case
\begin{equation*}
\E_X \left[\lambda_1^4 + \dots + \lambda_p^4 \right] = \Theta(1)
\end{equation*}
and 
\begin{equation*}
\var\left(\tr (\Sigma \hat{\Sigma}^{-1})\right) = \Theta(1/n).
\end{equation*}

\end{proof}

\section{Connection with results about the k-fold test error}
\label{app:kfold_test_error}

Another estimand considered in the cross-validation literature is the \emph{k-fold test error}: the average accuracy of the $k$ submodels fit during cross-validation~\cite[e.g.,][]{bayle2020}. Formally, this quantity is defined as 
\begin{equation*}
    \err_{\textnormal{k-fold}}= \frac{1}{K} \sum_{k = 1}^K \E \left[ \ell(\hat{f}(X_{n+1}, \hat{\theta}^{(-k)}), Y_{n+1}) \mid (X,Y) \right],
\end{equation*}
where the expectation is over a fresh test point $(X_{n+1}, Y_{n+1}) \sim P$ with the training data held fixed. This quantity has earned recognition because of its technical properties. For example,~\citet{bayle2020} prove a far-reaching central limit theorem for $(\errhat^\cv -  \err_{\textnormal{k-fold}})$. (Note that the expected value of this quantity is zero. One might say that $\errhat^\cv$ is unbiased for $ \err_{\textnormal{k-fold}}$, with the tacit understanding that the latter quantity is random.)

How does $\err_{\textnormal{k-fold}}$ relate to $\err$ and $\err_{XY}$? We offer some initial observations here. First, we carry out a numerical experiment, returning to the OLS setting of Figure~\ref{fig:ols_target}. Here, we find that cross-validation is a better estimate of $\err$ than of $\err_{\textnormal{k-fold}}$, and a worse estimate of $\err_{\textnormal{k-fold}}$ than of $\err_{XY}$. See Figure~\ref{fig:ols_target_kfold}. In this setting, the correlation between 
$\errhat^\cv$ and $\err_{XY}$ is less than $1\%$, whereas the correlation between $\errhat^\cv$ and $\err_{\textnormal{k-fold}}$ is $8.5\%$. That is, the CV point estimate is tracking the variation in $\err_{\textnormal{k-fold}}$ in a non-negligible way; see Figure~\ref{fig:ols_cor_kfold}. In view of the definition of the k-fold test error, we expect that this correlation is nontrivial asymptotically, but we set aside a full investigation of this asymptotic correlation to future work.

\begin{figure}
    \centering
    \includegraphics[height = 2.75in]{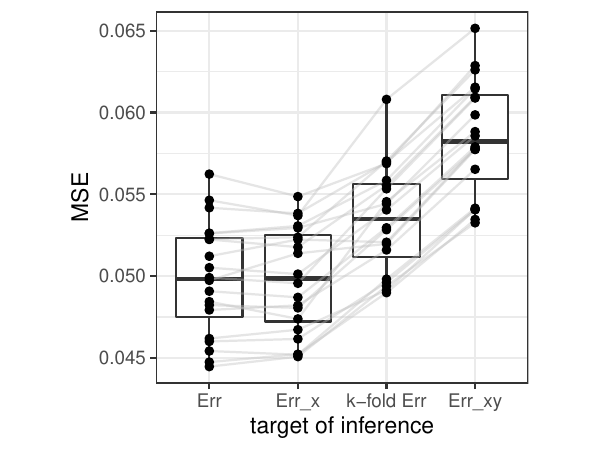}
    \caption{MSE of the cross-validation point estimate for estimating four targets of inference.}
    \label{fig:ols_target_kfold}
\end{figure}

Next, we consider the an estimator for the variance of $(\errhat^\cv -  \err_{\textnormal{k-fold}})$ from~\citet{bayle2020}. The first estimate is the \emph{all-pairs variance estimator},
\begin{equation*}
    \sehat^2_{\textnormal{out}} := \cdot \frac{1}{nK} \sum_{k=1}^K \frac{1}{n/K} \sum_{i \in \Ical_k} \left(\ell(\hat{f}(x_i), \hat{\theta}^{(-k)}, y_i) - \errhat^\cv\right)^2 = \frac{n-1}{n} \sehat^2. 
\end{equation*}
Notice this results in (very slightly) \emph{smaller} confidence intervals than the na\"ive intervals we consider in this work. Thus, the performance of the so-called na\"ive estimator in the simulations presented in this work is essentially the coverage of the CIs created using the central limit theorem of~\cite{bayle2020} together with this variance estimator (the coverage of the na\"ive estimator in this work is slightly larger). Since these CLT intervals are shown to be asymptotically exact for the quantity $\err_{\textnormal{k-fold}}$, the poor coverage that we observe in our simulations is due either to (i) the limiting behavior derived in \citet{bayle2020} is not a good approximation at the modest sample sizes we consider (ii) there is different precision when estimating $\err_{XY}$ compared to $\err_{\textnormal{k-fold}}$. Regarding the first point, that work introduced algorithmic stability conditions that may not be satisfied in the high-dimensional or low sample size settings we consider.

\begin{figure}[!h]
\captionsetup[subfigure]{justification=centering}
\centering
\begin{subfigure}[b]{.49\textwidth}
\centering
\includegraphics[height = 2.5in, trim = 0 0 0 0, clip]{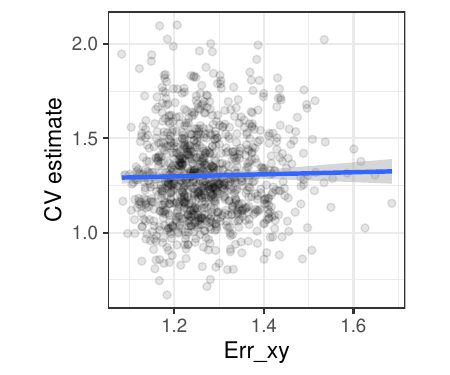}
\end{subfigure}
\begin{subfigure}[b]{.49\textwidth}
\centering
\includegraphics[height = 2.5in, trim = 0 0 0 0, clip]{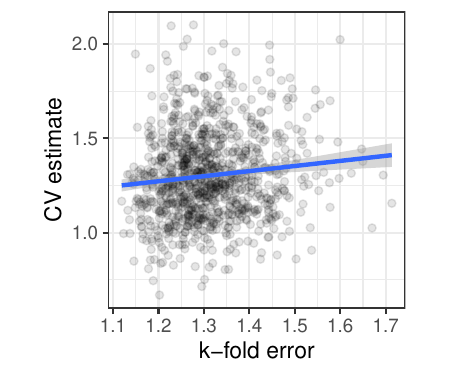}
\end{subfigure}
\caption{Scatter plot showing the correlation of the CV point estimate with $\err_{XY}$ (left) and $\err_{\textnormal{k-fold}}$ (right). The blue line is the best fit regression.}
\label{fig:ols_cor_kfold}
\end{figure}


\end{document}